\documentclass{article}
\usepackage{fullpage}
\usepackage{amsfonts, amsmath, amssymb, amsthm} 
\usepackage{authblk}
\usepackage{tikz} 

\usepackage{mathtools, eqparbox}%

\usepackage{wrapfig}
\usepackage[short,nocomma]{optidef}
\usepackage{cases}
\usepackage{todonotes}
\usepackage{caption}
\usepackage{subcaption}
\usepackage{nicefrac} 
\usepackage{xspace,enumerate}
\usepackage[utf8]{inputenc}
\usepackage{thmtools}
\usepackage{thm-restate}
\usepackage{hyperref}
\usepackage[capitalise]{cleveref}
\usepackage{verbatim}
\usepackage{multirow}
\usepackage{url}
\usepackage{mathtools}
\usepackage{color}
\usepackage{scalerel}
\usepackage{algorithm}
\usepackage[noend]{algpseudocode}
\usepackage[normalem]{ulem}
\usetikzlibrary{calc,snakes,shapes,arrows.meta}
\usepackage{booktabs}
\usepackage{graphicx}

\newcommand{\defproblem}[3]{
  \vspace{2mm}
  \begin{center}
  \noindent\fbox{
  \begin{minipage}{0.96\textwidth}
  \textsc{#1}

  \smallskip
  \noindent
  {\bf{Input:}} #2
  
  \smallskip
  \noindent
  {\bf{Output:}} #3
  \end{minipage}
  }
  \end{center}
  \vspace{2mm}
}

\def\dd{\mathinner{.\,.}} 
\newcommand{\cO}{\mathcal{O}}
\newcommand{\ctO}{\mathcal{\tilde{O}}}

\newcommand{\EDSI}{\text{EDSI}\xspace}
\newcommand{\LCP}{\textsf{LCP}\xspace}

\renewcommand{\L}{\mathcal{L}}
\newcommand{\NN}{\mathbb{N}}

\theoremstyle{plain}
\newtheorem{theorem}{Theorem}[section]
\newtheorem{fact}[theorem]{Fact}
\newtheorem{lemma}[theorem]{Lemma}

\newtheorem{proposition}[theorem]{Proposition}
\newtheorem{corollary}[theorem]{Corollary}

\newtheorem{conjecture}[theorem]{Conjecture}

\theoremstyle{remark}
\newtheorem{definition}[theorem]{Definition}

\newtheorem{example}[theorem]{Example}

\title{Elastic-Degenerate String Comparison\thanks{ This work was partially supported by the PANGAIA, ALPACA and NETWORKS projects that have received funding from the European Union's Horizon 2020 research and innovation programme under the Marie Skłodowska-Curie grant agreements No. 872539, 956229 and 101034253, respectively. Nadia Pisanti was partially supported by MUR PRIN 2022 YRB97K PINC and by NextGeneration EU programme PNRR ECS00000017 Tuscany Health Ecosystem. Jakub Radoszewski was supported by the Polish National Science Center, grants no.\ 2018/31/D/ST6/03991 and 2022/46/E/ST6/00463.}}

\author[1]{Esteban Gabory}
\author[2]{Moses Njagi Mwaniki}
\author[2]{Nadia Pisanti}
\author[1,3]{Solon P.\ Pissis}
\author[4]{Jakub Radoszewski}
\author[1]{Michelle Sweering}
\author[1,4]{Wiktor Zuba}

\affil[1]{CWI, Amsterdam, The Netherlands}
\affil[2]{University of Pisa, Pisa, Italy}
\affil[3]{Vrije Universiteit, Amsterdam, The Netherlands}
\affil[4]{Institute of Informatics, University of Warsaw, Warsaw, Poland}

\date{\today}

\begin{document}

\maketitle

\begin{abstract}
An elastic-degenerate (ED) string $T$ is a sequence of $n$ sets $T[1],\ldots,T[n]$ containing $m$ strings in total whose cumulative length is $N$. We call $n$, $m$, and $N$ the length, the cardinality and the size of $T$, respectively.
The language of $T$ is defined as $\L(T)=\{S_1 \cdots S_n\,:\,S_i \in T[i]\text{ for all }i\in[1,n]\}$.
ED strings have been introduced to represent a set of closely-related DNA sequences, also known as a pangenome. The basic question we investigate here is: Given two ED strings, how fast can we check whether the two languages they represent have a nonempty intersection? We call the underlying problem the \textsc{ED String Intersection} (EDSI) problem.
For two ED strings $T_1$ and $T_2$ of lengths $n_1$ and $n_2$, cardinalities $m_1$ and $m_2$, and sizes $N_1$  and $N_2$, respectively, we show the following:
\begin{itemize}
    \item There is no $\cO((N_1N_2)^{1-\epsilon})$-time algorithm, thus no $\cO\left((N_1m_2+N_2m_1)^{1-\epsilon}\right)$-time algorithm and no $\cO\left((N_1n_2+N_2n_1)^{1-\epsilon}\right)$-time algorithm, for any constant $\epsilon>0$, for \EDSI even when $T_1$ and $T_2$ are over a binary alphabet, unless the Strong Exponential-Time Hypothesis is false.
    \item There is no combinatorial $\cO((N_1+N_2)^{1.2-\epsilon}f(n_1,n_2))$-time algorithm, for any constant $\epsilon>0$ and any function $f$, for \EDSI even when $T_1$ and $T_2$ are over a binary alphabet, unless the Boolean Matrix Multiplication conjecture is false.
    \item An $\cO(N_1\log N_1\log n_1+N_2\log N_2\log n_2)$-time algorithm for outputting a compact (RLE) representation of the intersection language of two unary ED strings. In the case when $T_1$ and $T_2$ are given in a compact representation, we show that the problem is NP-complete.
    \item An $\cO(N_1m_2+N_2m_1)$-time algorithm for \EDSI.
    \item An $\ctO(N_1^{\omega-1}n_2+N_2^{\omega-1}n_1)$-time algorithm for \EDSI, where $\omega$ is the exponent of matrix multiplication; the $\ctO$ notation suppresses factors that are polylogarithmic in the input size.
\end{itemize}

We also show that the techniques we develop here have many applications even outside of bioinformatics. 
\end{abstract}


\section{Introduction}\label{sec:intro}

Sequence (or string) comparison is a fundamental task in computer science, with numerous applications in computational biology~\cite{DBLP:books/cu/Gusfield1997}, signal processing~\cite{Signal1979}, information retrieval~\cite{DBLP:books/aw/Baeza-YatesR2011}, file comparison~\cite{DBLP:journals/cacm/Heckel78}, pattern recognition~\cite{DBLP:journals/prl/AyadBP17}, security~\cite{DBLP:journals/ipl/ManberW94}, and elsewhere~\cite{DBLP:journals/csur/Navarro01}. 
Given two or more sequences and a distance function, the task is to compare the sequences in order to infer or visualize their (dis)similarities~\cite{DBLP:books/daglib/0020103}.

Many sequence representations have been introduced over the years to account for \emph{unknown} or \emph{uncertain} letters, a phenomenon that often occurs in data that comes from experiments~\cite{DBLP:conf/icdm/0001CGGLPPPSS20}. In the context of computational biology, for example, the IUPAC notation~\cite{IUPAC} is used to represent loci in a DNA sequence for which several alternative nucleotides are possible as variants. This gives rise to the notion of \emph{degenerate string} (or \emph{indeterminate string}): a sequence of finite sets of \emph{letters}~\cite{DBLP:journals/fuin/AlzamelABGIPPR20}. When all sets are of size 1, we are in the special case of a \emph{standard string} (or \emph{deterministic string}). Degenerate strings can encode the consensus of a population of DNA sequences~\cite{consensus2021} in a gapless multiple sequence alignment (MSA). Iliopoulos et al.~generalized this notion to also encode insertions and deletions (gaps) occurring in MSAs by introducing the notion of \emph{elastic-degenerate string}: a sequence of finite sets of \emph{strings}~\cite{DBLP:journals/iandc/IliopoulosKP21}.

The main motivation to consider elastic-degenerate (ED) strings is that they can be used to represent a \emph{pangenome}: a \emph{collection} of closely-related genomic sequences that are meant to be analyzed together~\cite{PanGenomeConsortium18}. Several other, more powerful, pangenome representations have been proposed in the literature, mostly graph-based ones; see the comprehensive survey by Carletti et al.~\cite{DBLP:conf/gbrpr/CarlettiFG0RV19} or by Baaijens et al.~\cite{DBLP:journals/nc/BaaijensBBVPRS22}. Compared to these more powerful representations, ED strings have algorithmic advantages, as they support: (i) fast and simple on-line string matching~\cite{DBLP:conf/cpm/GrossiILPPRRVV17,DBLP:journals/bioinformatics/CislakGH18}; (ii) (deterministic) subquadratic string matching~\cite{DBLP:conf/cpm/AoyamaNIIBT18,bernardini_et_al:LIPIcs:2019:10597,elasticSICOMP}; and (iii) efficient approximate string matching~\cite{tcs-ed2020,DBLP:conf/latin/0001GPSSZ22}.

Our main goal here is to give the first algorithms and lower bounds for comparing two pangenomes represented by two ED strings.\footnote{Pangenome comparison is one of the central goals of two large EU funded projects on computational pangenomics: PANGAIA (\url{https://www.pangenome.eu/}) and ALPACA (\url{https://alpaca-itn.eu/}).} We consider the most basic notion of matching, namely, to decide whether two ED strings, each encoding a language, have a nonempty intersection. Like with standard strings, algorithms for pairwise ED string comparison can serve as computational primitives for many analysis tasks (e.g., phylogeny reconstruction); lower bounds for pairwise ED string comparison can serve as meaningful lower bounds for the comparison of more powerful pangenome representations such as, for instance, variation graphs~\cite{DBLP:conf/gbrpr/CarlettiFG0RV19}.

Let us start with some basic definitions and notation.
An \emph{alphabet} $\Sigma$ is a finite nonempty set of elements called \emph{letters}. By $\Sigma^*$ we denote the set of all strings over $\Sigma$ including the \emph{empty string} $\varepsilon$ of length $0$.
For a string $S=S[1] \cdots S[n]$ over $\Sigma$, we call $n=|S|$ its \emph{length}. The fragment $S[i\dd j]$ of $S$ is an \emph{occurrence} of the underlying \emph{substring} $P=S[i]\cdots S[j]$. 
We also say that $P$ occurs at \emph{position} $i$ in $S$. 
A {\em prefix} of $S$ is a fragment of $S$ of the form $S[1\dd j]$ and a {\em suffix} of $S$ is a fragment of $S$ of the form $S[i\dd n]$. 
An \emph{elastic-degenerate string} (ED string, in short) $T$ is a sequence $T=T[1] \cdots T[n]$ of $n$ finite sets, where $T[i]$ is a subset of $\Sigma^*$. The total size of $T$ is defined as $N=N_{\varepsilon}+\sum^{n}_{i=1}\sum_{S\in T[i]} |S|$, where $N_{\varepsilon}$ is the total number of empty strings in $T$. By $m$ we denote the total number of strings in all $T[i]$, i.e., $m=\sum_{i=1}^n |T[i]|$. We say that $T$ has \emph{length} $n=|T|$, \emph{cardinality} $m$ and \emph{size} $N=||T||$. An ED string $T$ can be treated as a compacted nondeterministic finite automaton (NFA) with $n+1$ states, called \emph{segments}, numbered $1,\ldots,n+1$, and $m$ transitions labeled by strings in $\Sigma^*$. State $1$ is \emph{starting} and state $n+1$ is \emph{accepting}. For each index $i \in [1,n]$ and string $S \in T[i]$, there is a transition from state $i$ to state $i+1$ with label $S$; inspect also~\cref{fig:NFA} for an example. 
The language $\L(T)$ generated by the ED string $T$ is the language accepted by this compacted NFA. That is, $\L(T)=\{S_1 \cdots S_n\,:\,S_i \in T[i]\text{ for all }i\in[1,n]\}$. 

\begin{figure}[htpb]
    \centering
    \includegraphics[width=11cm]{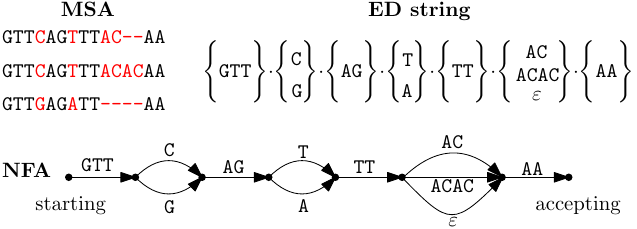}
    \caption{An example of an MSA and its corresponding (non-unique) ED string $T$ of length $n=7$, cardinality $m=11$ and size $N=20$, and the compacted NFA for $T$. The compacted NFA can be seen as a special case of an edge-labeled directed acyclic graph.}
    \label{fig:NFA}
\end{figure}

We next define the main problem in scope; inspect also Figure~\ref{fig:EDS} for an example.

\defproblem{ED String Intersection (\EDSI)}{Two ED strings, $T_1$ of length $n_1$, cardinality $m_1$ and size $N_1$, and $T_2$ of length $n_2$, cardinality $m_2$ and size $N_2$.}{YES if $\L(T_1)$ and $\L(T_2)$ have a nonempty intersection, NO otherwise.}

\begin{figure}[htpb]
    \centering
    \includegraphics[width=11cm]{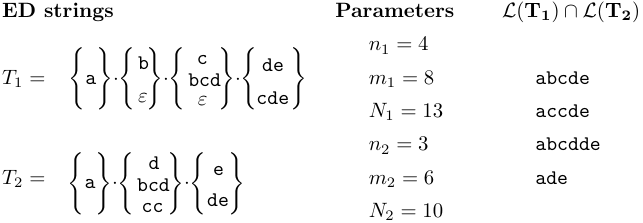}
    \caption{An example of two ED strings $T_1$ and $T_2$ with their parameters and the intersection of their languages. In this instance, we see that $\L(T_1)$ and $\L(T_2)$ have a nonempty intersection.}
    \label{fig:EDS}
\end{figure}

\subparagraph{Our Results} We make the following specific contributions:

\begin{enumerate}
     \item In Section~\ref{sec:SETH}, we give several conditional lower bounds. In particular, we show that there is no $\cO((N_1N_2)^{1-\epsilon})$-time algorithm, thus no $\cO\left((N_1m_2+N_2m_1)^{1-\epsilon}\right)$-time algorithm and no $\cO\left((N_1n_2+N_2n_1)^{1-\epsilon}\right)$-time algorithm, for any constant $\epsilon>0$, for \EDSI even when $T_1$ and $T_2$ are over a binary alphabet, unless the Strong Exponential-Time Hypothesis (SETH)~\cite{DBLP:journals/jcss/ImpagliazzoP01,DBLP:journals/jcss/ImpagliazzoPZ01} or the Orthogonal Vectors (OV) conjecture~\cite{DBLP:journals/tcs/Williams05} is false.
    \item In Section~\ref{sec:BMM}, we present other conditional lower bounds. In particular, we show that there is no combinatorial\footnote{The notion of ``combinatorial algorithm'' is informal but widely used in the literature. Typically, we call an algorithm ``combinatorial'' if it does not not call an oracle for ring matrix multiplication.} $\cO((N_1+N_2)^{1.2-\epsilon}f(n_1,n_2))$-time algorithm, for any constant $\epsilon>0$ and any function $f$, for \EDSI even when $T_1$ and $T_2$ are over a binary alphabet, unless the Boolean Matrix Multiplication (BMM) conjecture~\cite{DBLP:conf/focs/AbboudW14} is false.
    \item In Section~\ref{sec:unary}, we show an $\cO(N_1\log N_1\log n_1+N_2\log N_2\log n_2)$-time algorithm for outputting a compact (RLE) representation of the intersection language of two unary ED strings. In the case when $T_1$ and $T_2$ are given in a compact representation, we show that the problem is NP-complete.
    \item In Section~\ref{sec:graph}, we show an $\cO(N_1m_2+N_2m_1)$-time combinatorial algorithm for \EDSI in which we assume that the ED strings are over an integer alphabet $[1,(N_1+N_2)^{\cO(1)}]$.
    \item In Section~\ref{sec:MM}, we show an $\ctO(N_1^{\omega-1}n_2+N_2^{\omega-1}n_1)$-time algorithm for \EDSI, where $\omega$ is the matrix multiplication exponent.
\end{enumerate}

Interestingly, we show that the techniques we develop here have applications outside of bioinformatics. Given a sequence $P = P_1,\ldots,P_n$ of $n$ standard strings, we define an \emph{acronym} of $P$ as a string $A=A_1\cdots A_n$, where $A_i$ is a possibly empty prefix of $P_i$, for all $i\in[1,n]$. In the \textsc{Acronym Generation} (AG) problem, we are given a dictionary $D$ of $k$ strings of total length $K$ and a sequence $P$ of $n$ strings of total length $N$, and we are asked to say YES if and only if some acronym of $P$ belongs to $D$. In Section~\ref{sec:AG}, we show how our techniques for \EDSI can be modified to solve AG in $\cO(nK+N)$ time. 

In \cref{sec:EDapps}, we show how intersection graphs can be used to solve different ED string comparison tasks.
In \cref{sec:DEDSM}, we show how intersection graphs can be used for string matching in the general case when both the pattern and the text are ED strings.
In \cref{sec:AEDSI}, we extend our results for EDSI to the approximate case (under the Hamming or edit distance). We conclude this paper in \cref{sec:fin} with some open problems.

This is a full and extended version of a conference paper~\cite{DBLP:conf/cpm/GaboryMPPRSZ23}.

\subparagraph{Related Work} Apart from its applications to pangenome comparison, \EDSI is interesting theoretically on its own as a special case of regular expression (regex) matching. Regex is a basic notion in formal languages and automata theory. Regular expressions are commonly used in practical applications to define search patterns. Regex matching and membership testing are widely used as computational primitives in programming languages and text processing utilities (e.g., the widely-used \texttt{agrep}). The classic algorithm for solving these problems constructs and simulates an NFA corresponding to the regex, which gives an $\cO(MN)$ running time, where $M$ is the length of the pattern and $N$ is the length of the text. Unfortunately, significantly faster solutions are unknown and unlikely~\cite{DBLP:conf/focs/BackursI16}. However, much faster algorithms exist for many special cases of the problem: dictionary matching, wildcard matching, subset matching, and the word break problem (see~\cite{DBLP:conf/focs/BackursI16} and references therein) as well as for sparse regex matching~\cite{DBLP:conf/soda/BilleG24}.


Special cases of \EDSI have also been studied.
First, let us consider the case when both $T_1$ and $T_2$ are degenerate strings. In this case, the problem is trivial: \EDSI has a positive answer if and only if for every $i$, $T_1[i] \cap T_2[i]$ is nonempty. Alzamel et al.~\cite{wabi18,DBLP:journals/fuin/AlzamelABGIPPR20} studied the case when $T_1$ and $T_2$ are \emph{generalized degenerate strings}: for any $i\in[1,n_1]$ and $j\in[1,n_2]$ all strings in $T_1[i]$ have the same length $\ell_{1,i}$ and all strings in $T_2[j]$ have the same length $\ell_{2,j}$. In the case of generalized degenerate strings, they showed that deciding if $\L(T_1)$ and $\L(T_2)$ have a nonempty intersection can be done in $\cO(N_1+N_2)$ time. If $T_2$ is a standard string, i.e., an ED string with $m_2=n_2=1$, then we can resort to the results of Bernardini et al.~\cite{elasticSICOMP} for ED string matching. In particular: there is no combinatorial algorithm for \EDSI working in $\cO(n_1N_2^{1.5-\epsilon}+N_1)$ time unless the BMM conjecture is false; and we can solve \EDSI in $\ctO(n_1N_2^{\omega-1}+N_1)$ time. Moreover, Gawrychowski et al.~\cite{DBLP:conf/cpm/GawrychowskiGL20} provided a systematic study of the complexity of degenerate string comparison under different notions of matching: Cartesian tree matching; order-preserving matching; and parameterized matching.

Similar to ED strings (and to generalized degenerate strings) is the representation of pangenomes via \emph{founder graphs}. The idea behind founder graphs is that a multiple alignment of few \emph{founder sequences} can be used to approximate the input MSA, with the feature that each row of the MSA is a recombination of the founders. Like founder graphs, ED strings support the recombination of different rows of the MSA between consecutive columns. Unlike ED strings, for which no efficient index is probable~\cite{DBLP:conf/spire/Gibney20} (and indeed their value is to enable fast on-line string matching), some subclasses of founder graphs are indexable, and a recent research line is devoted to constructing and indexing such structures~\cite{DBLP:conf/wabi/Ascone0CEGGP24,DBLP:journals/algorithmica/EquiNACTM23,makinen2020linear,DBLP:journals/tcs/RizzoENM24}. In general, both ED strings and founder graphs are special cases of labeled graphs. Unfortunately, indexing labeled graphs is unlikely to have an efficient solution~\cite{DBLP:journals/talg/EquiMTG23}.

\section{Conditional Lower Bounds}\label{sec:CLB}

In this section, we show several conditional lower bounds for the \EDSI problem.
Bounds in the first group (see Section~\ref{sec:SETH}) are based on the popular Strong Exponential-Time Hypothesis (SETH)~\cite{DBLP:conf/iwpec/CalabroIP09}; the second group of bounds (see Section~\ref{sec:BMM}) is based, instead, on the Boolean Matrix Multiplication (BMM) conjecture~\cite{DBLP:conf/focs/AbboudW14}.

\subsection{Lower Bounds Based on SETH}\label{sec:SETH}

We are going to reduce the \textsc{Orthogonal Vectors} (OV, in short) problem to the \EDSI problem. In the OV problem we are given a set $V = \{v^1,\ldots,v^k\}$ of $k$ binary vectors, each of length $d$, and we are asked to decide whether or not there are any two vectors in $V$  which are orthogonal; i.e., the dot product of the two vectors is zero. The OV conjecture, implied by SETH (see~\cite{DBLP:journals/tcs/Williams05}), is the following.
\begin{conjecture}[OV conjecture~\cite{DBLP:journals/tcs/Williams05}]
The OV problem for $k$ binary vectors, each of length $d=\Theta(\log k)$, cannot be solved in $\cO(k^{2-\epsilon})$ time, for any constant $\epsilon>0$. 
\end{conjecture}

We show the following reduction.

\begin{theorem}
Given any set $V=\{v^1,\ldots, v^k\}$ of $k$ binary vectors of length $d$, we can construct in linear time two ED strings $T_1$ and $T_2$ over a binary alphabet such that:
\begin{itemize}
    \item $T_1$ has length, cardinality, and size $\Theta(kd)$;
    \item $T_2$ has length $\Theta(\log k)$, cardinality $\Theta(k)$ and size $\Theta(kd)$; and
    \item $V$ contains two orthogonal vectors if and only if $T_1$ and $T_2$ have a nonempty intersection.
\end{itemize}
\end{theorem}
\begin{proof}
Let $u^i = 1^d - v^i$ for all $i \in [1,k]$. For a length-$d$ vector $v$ and $j \in [1,d]$, by $v_j$ we denote the $j$th component of $v$. We construct $T_1$ and $T_2$ as follows~(see Example~\ref{ex:OVreduction}):\footnote{By the $\prod$ notation we denote a sequence of concatenations of segments in an ED string.}

\begin{align*}
    T_1 = \prod_{i=1}^k\ \prod_{j=1}^d\ \left\{0, u^i_j\right\},\quad\quad\quad
    T_2 = \prod_{i=0}^{\lfloor\log_2 k\rfloor}\left\{0^{d\cdot2^i}, \varepsilon\right\}\cdot V\cdot \prod_{i=0}^{\lfloor\log_2 k\rfloor}\left\{0^{d\cdot2^i}, \varepsilon\right\}.
\end{align*}

We now show that $\mathcal{L}(T_1)$ and $\mathcal{L}(T_2)$ have a nonempty intersection if and only if there exists a pair of orthogonal vectors in $V$.
\begin{itemize}
    \item Suppose $v^a$ and $v^b$ are orthogonal. Then for all $j \in [1,d]$, $v^b_j \in \left\{0, u^a_j\right\}$ and hence $v^b \in \prod_j \left\{0, u^a_j\right\}$. It follows that \[0^{(a-1)d}v^b0^{(k-a)d} \in 0^{(a-1)d}\prod_j \left\{0, u^a_j\right\}0^{(k-a)d}\subseteq \mathcal{L}(T_1).\]
    By decomposing $a-1 = \sum_{i \in S_{a-1}} 2^i$ and $k-a = \sum_{i \in S_{k-a}} 2^i$, where for any integer $p$, the set $S_p$ contains the positions with a $1$ in the binary representation of $p$, we find that
    \[0^{(a-1)d}v^b0^{(k-a)d} \in \prod_{i\in S_{a-1}} 0^{d\cdot2^i} \cdot V\cdot \prod_{i \in S_{k-a}} 0^{d\cdot2^i} \subseteq \mathcal{L}(T_2).\]
    We conclude that $0^{(a-1)d}v^b0^{(k-a)d} \in \mathcal{L}(T_1) \cap \mathcal{L}(T_2)$.
    \item Conversely, suppose that $\mathcal{L}(T_1)$ and $\mathcal{L}(T_2)$ have a nonempty intersection and consider a string $S\in \mathcal{L}(T_1) \cap \mathcal{L}(T_2)$. Let $v^b$ be the vector from $V$ which is chosen in $T_2$ when constructing $S$. The strings in the sets of $T_2$ all have length divisible by $d$. Thus $v^b$ starts at an index $(a-1)d + 1$ of string $S$ for some integer $a$. Since $S\in \mathcal{L}(T_1)$, we have $v^b \in \prod_{j=1}^d\ \left\{0, u^a_j\right\}$. This implies that $v^a$ and $v^b$ are orthogonal.
\end{itemize}
Therefore, solving the orthogonal vectors problem for $V$ is equivalent to checking whether $\mathcal{L}(T_1)$ and $\mathcal{L}(T_2)$ have a nonempty intersection.
\end{proof}

\begin{example}\label{ex:OVreduction}
Let $k=3$, $d=2$ and $V=
\{
v^1=(1,0),
v^2=(0,1),
v^3=(1,1)\}$.

We have that $T_1=
\begin{Bmatrix}0\end{Bmatrix}\cdot
\begin{Bmatrix}0\\1\end{Bmatrix}\cdot
\begin{Bmatrix}0\\1\end{Bmatrix}\cdot
\begin{Bmatrix}0\end{Bmatrix}\cdot
\begin{Bmatrix}0\end{Bmatrix}\cdot
\begin{Bmatrix}0\end{Bmatrix}$

and $T_2=
\begin{Bmatrix}00\\\varepsilon\end{Bmatrix}\cdot
\begin{Bmatrix}0000\\\varepsilon\end{Bmatrix}\cdot
\begin{Bmatrix}10\\01\\11\end{Bmatrix}\cdot
\begin{Bmatrix}00\\\varepsilon\end{Bmatrix}\cdot
\begin{Bmatrix}0000\\\varepsilon\end{Bmatrix}
$.

One can observe that each string from $\mathcal{L}(T_1)\cap\mathcal{L}(T_2)$ corresponds to a pair of orthogonal vectors from $V$. For example, the string $010000$ is in $\mathcal{L}(T_2)$ because $v^2=(0,1)\in V$. Since the vector $v^1=(1,0)\in V$ is orthogonal to $v^2$, one also has $010000\in\mathcal{L}(T_1)$. This is because the two first segments of $T_1$ are constructed to encode any vector which is orthogonal to $v^1$.
\end{example}

Note that when $d=\Theta(\log k)$, the length $n_1$, the cardinality $m_1$ and the size $N_1$ of $T_1$ are $\cO(k\log k)$, whereas $T_2$ has length $n_2=\cO(\log k)$, cardinality $m_2=\cO(k)$ and size $N_2=\cO(k\log k)$. Moreover, both ED strings are over a binary alphabet $\Sigma = \{0, 1\}$. This implies various hardness results for \EDSI. For example, we can see that, for any constant $\epsilon > 0$, and an alphabet $\Sigma$ of size at least $2$ the problem cannot be solved in 
\[\cO\left((N_1 + N_2 + n_1 + n_2)^{2-\epsilon}\cdot \text{poly}(n_2)\right)\]
time, conditional on the OV conjecture.
By using the fact that $n_1\leq m_1 \leq N_1$ and $n_2\leq m_2 \leq N_2$, we obtain the following bounds.

\begin{corollary}
For any constant $\epsilon > 0$, there exists no
\begin{itemize}
    \item $\cO((N_1N_2)^{1-\epsilon})$-time
    \item $\cO\left((N_1m_2+N_2m_1)^{1-\epsilon}\right)$-time
    \item $\cO\left((N_1n_2+N_2n_1)^{1-\epsilon}\right)$-time
\end{itemize}
algorithm for the \EDSI problem, unless the OV conjecture is false.
\end{corollary}

\subsection{Combinatorial Lower Bounds Based on BMM Conjecture}\label{sec:BMM}

In the \textsc{Triangle Detection} (TD, in short) problem we are given three $D \times D$ Boolean matrices $A$, $B$, $C$ and one has to check whether there are three indices $i,j,k \in [0,D)$ such that $A[i,j]=B[j,k]=C[k,i]=1$.
It is known that Boolean Matrix Multiplication (BMM) and TD either both have truly subcubic combinatorial algorithms or none of them do~\cite{DBLP:conf/focs/WilliamsW10}. The BMM conjecture is stated as follows.

\begin{conjecture}[BMM conjecture~\cite{DBLP:conf/focs/AbboudW14}]
Given two $D \times D$ Boolean matrices, there is no combinatorial algorithm for BMM working in $\cO(D^{3-\epsilon})$ time, for any constant $\epsilon>0$.
\end{conjecture}

Our reduction from TD to EDSI is based on the construction of Bernardini et al.~from~\cite{elasticSICOMP} for ED string matching. 

\begin{theorem}
If \EDSI over a binary alphabet can be solved in $\cO((N_1+N_2)^{1.2-\epsilon}f(n_1,n_2))$ time, for any constant $\epsilon>0$ and any function $f$, then there exists a truly subcubic combinatorial algorithm for TD.
\end{theorem}
\begin{proof}

Let $D$ be a positive integer and let $A$, $B$, and $C$ be three $D\times D$ Boolean matrices. Further let $s\in [1,D]$ be an integer to be set later. In the rest of the proof, we can assume that $s$ divides $D$, up to adding $\alpha$ rows and columns containing only $0$'s to all three matrices, where $\alpha$ is the smallest non-negative representative of the equivalence class $-D\bmod{s}$.

Let us first construct an ED string $T_1=X_1X_2X_3$ over a large alphabet with $n_1=3$, where each $X_p$, $p\in[1,n_1]$, contains a string for each occurrence of value $1$ in $A$, $B$ and $C$, respectively. Below $i$ iterates over $[0,D)$, $j$ and $k$ iterate over $[0,\frac{D}{s})$, and $x$ and $y$ iterate over $[0,s)$. Moreover, $x,y\in[0,s)$, $v_i$ for $i \in [0,D)$, and $a$, $\$$ are all letters.
\begin{itemize}
    \item If $A[i,x\cdot \frac{D}{s}+j]=1$, then $X_1$ contains the string $v_ixa^j$.
    \item If $B[x\cdot \frac{D}{s}+j,y\cdot\frac{D}{s}+k]=1$, then $X_2$ contains the string $a^{\frac{D}{s}-j}x\$\$ya^{\frac{D}{s}-k}$.
    \item If $C[y\cdot\frac{D}{s}+k,i]=1$, then $X_3$ contains the string $a^kyv_i$.
\end{itemize}
The length of each string in each $X_p$ is $\cO(D/s)$ and the total number $m_1$ of strings is up to $3D^2$. Overall, $N_1=\cO(D^3/s)$.

We construct an ED string $T_2$ with $n_2=1$ containing the following strings:
\[P(i,x,y)=v_ixa^{\frac{D}{s}}x\$\$ya^{\frac{D}{s}}yv_i\quad\text{for every }x,y \in [0,s)\text{ and }i\in[0,D).\]
Each string has length $\cO(D/s)$ and there are $m_2=Ds^2$ strings, so $N_2=\cO(D^2s)$.

We use the following fact.
\begin{fact}[\cite{elasticSICOMP}]
$P(i,x,y)\in \mathcal{L}(T_1)$ if and only if the following holds for some $j, k \in [0,D /s)$:
\[A[i, x \cdot \tfrac{D}{s} + j] = B[x \cdot \tfrac{D}{s} + j, y \cdot \tfrac{D}{s} + k] = C[y \cdot \tfrac{D}{s} + k, i] = 1.\]
\end{fact}

We choose $s=\lfloor \sqrt{D} \rfloor$; then $N_1,N_2=\cO(D^{2.5})$ and $n_1,n_2=\cO(1)$. Then indeed an $\cO((N_1+N_2)^{1.2-\epsilon}f(n_1,n_2))$-time algorithm for \EDSI would yield an $\cO(D^{3-2.5\epsilon})$-time algorithm for the TD problem.

Note also that even though the size of the alphabet used above is $\Theta(s+D)=\Theta(D)$, we can encode all letters by equal-length binary strings blowing $N_1$ and $N_2$ up only by a factor of $\Theta(\log D)$ and, hence, obtain the same lower bound for a binary alphabet.
\end{proof}

Both $m_1$ and $m_2$ in the reduction are $\cO(D^2)$, thus an $\cO((m_1+m_2)^{1.5-\epsilon}f(n_1,n_2))$-time algorithm would yield an $\cO(D^{3-2\epsilon})$-time algorithm for TD. We obtain the following.

\begin{corollary}
If \EDSI over a binary alphabet can be solved in $\cO((N_1^{1.2}+N_2^{1.2}+m_1^{1.5}+m_2^{1.5})^{1-\epsilon}f(n_1,n_2))$ time, for any constant $\epsilon>0$ and any function $f$, then there exists a truly subcubic combinatorial algorithm for TD. 
\end{corollary}

\section{\EDSI: The Unary Case}\label{sec:unary}

An ED string is called \emph{unary} if it is over an alphabet of size 1. In this special case, if both $T_1$ and $T_2$ are over the same alphabet $\Sigma=\{a\}$, \EDSI boils down to checking whether there exists any $b \ge 0$ such that $a^b$ belongs to both $\L(T_1)$ and $\L(T_2)$.

Let $T$ be a unary ED string of length $n$ over alphabet $\Sigma=\{a\}$.
We define the \emph{compact representation} $R(T)$ of $T$ as the following sequence of sets of integers: \[\forall i\in[1,n]\ R(T)[i]=\{b_{i,1},b_{i,2},\dots,b_{i,m_i}\} \iff T[i]=\{a^{b_{i,1}},a^{b_{i,2}},\dots,a^{b_{i,m_i}}\},\]

\noindent where $b_{i,j}\geq0$ for all $i\in[1,n]$ and $j\in[1,m_i]$, the cardinality of $T$ is $m=\sum_{i=1}^{n} m_i$, and its size is $N=N_{\varepsilon} +\sum_{i=1}^{n} \sum_{j=1}^{m_i} b_{i,j}$, where $N_{\varepsilon}$ is the total number of empty strings in $T$.

\begin{theorem}
If $T_1$ and $T_2$ are unary ED strings and each is given in a compact representation, the problem of deciding whether $\L(T_1) \cap \L(T_2)$ is nonempty is NP-complete.
\end{theorem}

\begin{proof}
The problem is clearly in NP, as it is enough to guess a single element for each set in both $T_1$ and $T_2$, and then simply check if the sums match in linear time.
We show the NP-hardness through a reduction from the \textsc{Subset Sum} problem, which takes $n$ integers $b_1,b_2,\dots,b_n$ and an integer $c$, and asks whether there exist $x_i\in\{0,1\}$, for all $i\in[1,n]$, such that $\sum_{i=1}^n x_ib_i = c$. \textsc{Subset Sum} is NP-complete~\cite{DBLP:books/daglib/0015106} also for non-negative integers. For any instance of \textsc{Subset Sum}, we set $R(T_1)[i]=\{b_i,0\}$ for all $i\in[1,n]$, $n_2=1$ and $R(T_2)[1]=\{c\}$. Then the answer to the \textsc{Subset Sum} instance is YES if and only if $\L(T_1)\cap\L(T_2)$ is nonempty.
\end{proof}

In what follows, we provide an algorithm which runs in polynomial time in the size of the two unary ED strings when they are given uncompacted.

The set $\L(T)$ can be represented as a set $L(T) \subset \NN$ such that $\L(T)=\{a^\ell\,:\,\ell\in L(T)\}$. The set $L(T)$ will be stored as a sorted list (without repetitions). We will show how to efficiently compute $L(T_1)$ and $L(T_2)$. Then one can compute $L(T_1) \cap L(T_2)$ in $\cO(N_1+N_2)$ time, which allows, in particular, to check if $L(T_1) \cap L(T_2) = \emptyset$ (which is equivalent to $\L(T_1) \cap \L(T_2) = \emptyset$).

We show the computation for $L(T_1)$. The workhorse is an algorithm from the following \cref{lem:concat} that allows to compute the set $L(X_1X_2)$ of concatenation of two ED strings based on their sets $L(X_1),L(X_2)$.

\begin{lemma}\label{lem:concat}
Let $X_1$ and $X_2$ be ED strings. 
Given $L(X_1)$ and $L(X_2)$ such that $t_1=\max L(X_1)$ and $t_2=\max L(X_2)$, we can compute $L(X_1X_2)$ in $\cO((t_1+t_2) \log (t_1+t_2))$ time.
\end{lemma}
\begin{proof}
For two sets $A,B\subset \NN$, by $A+B$ we denote the set $\{a+b\,:\,a \in A,\,b \in B\}$. We then have $L(X_1X_2)=L(X_1)+L(X_2)$. 
Fast Fourier Transform (FFT)~\cite{DBLP:books/daglib/0023376} can be used directly to compute $L(X_1)+L(X_2)$ in $\cO((t_1+t_2) \log (t_1+t_2))$ time. 
\end{proof}

\begin{lemma}\label{lem:LT}
$L(T_1)$ can be computed in $\cO(N_1 \log N_1 \log n_1)$ time.
\end{lemma}
\begin{proof}
We apply the recursive algorithm described in Algorithm~\ref{algo:unary} to $T_1$.

\begin{algorithm}[htpb]
\begin{algorithmic}
\If{$k=1$}
\State Compute $L(T[1])$ na\"{\i}vely
\EndIf
\State{$i\gets\lfloor k/2 \rfloor$}
\State{$L_1\gets $Compute-$L(T[1] \cdots T[i])$}
\State{$L_2\gets$Compute-$L(T[i+1] \cdots T[k])$}
\State{\textbf{return} $L_1 + L_2$}
\caption{Compute-$L(T[1] \cdots T[k])$}\label{algo:unary}
\end{algorithmic}
\end{algorithm}

Let $N_{1,i}=\sum_{x\in L(T_1)[i]}\max(1,x)$  and $t_{1,i}=\max L(T_1[i])$ for $i \in [1,n_1]$. Obviously, $t_{1,i} \le N_{1,i}$.

We analyze the complexity of the recursion by levels. For the bottom level, $L(T_1[i])$ can be computed in $\cO(N_{1,i})$ time for each $i \in [1,n_1]$, which sums up to $\cO(N_1)$. For the remaining levels, we notice that $\max L(T_1[i] \cdots T_1[j]) = t_{1,i}+\cdots+t_{1,j}$. On each level, the fragments of $T_1$ that are considered are disjoint. Thus, the complexity on each level via \cref{lem:concat} is $\cO((\sum_{i=1}^{n_1} t_{1,i}) \log (\sum_{i=1}^{n_1} t_{1,i})) = \cO(N_1 \log N_1)$. The number of levels of recursion is $\cO(\log n_1)$; the complexity follows.
\end{proof}

\begin{theorem}
If $T_1$ and $T_2$ are unary ED strings, then $L(T_1) \cap L(T_2)$ can be computed in $\cO(N_1\log N_1\log n_1+N_2\log N_2\log n_2)$ time.
\end{theorem}
\begin{proof}
We use \cref{lem:LT} to compute $L(T_1)$ and $L(T_2)$ in the required complexity. Then $L(T_1) \cap L(T_2)$ can be computed via bucket sort.
\end{proof}

\section{\EDSI: General Case}\label{sec:EDSI}
Assuming that the two ED strings, $T_1$ and $T_2$, of total size $N_1+N_2$  are over an integer alphabet $[1,(N_1+N_2)^{\cO(1)}]$, we can sort the suffixes of all strings in $T_1[i]$, for all $i\in[1,n_1]$, and the suffixes of all strings in $T_2[j]$, for all $j\in[1,n_2]$, in $\cO(N_1+N_2)$ time~\cite{DBLP:conf/focs/Farach97}.

By $\LCP(X,Y)$ let us denote the length of the longest common prefix of two strings $X$ and $Y$. Given a string $S$ over an integer alphabet, we can construct a data structure over $S$ in $\cO(|S|)$ time, so that
when $i,j\in[1,|S|]$ are given to us on-line, we can determine $\LCP(S[i\dd |S|],S[j\dd |S|])$ in $\cO(1)$ time~\cite{DBLP:conf/latin/BenderF00}.

\subsection{Compacted NFA Intersection}\label{sec:graph}

In this section we show an algorithm for computing a representation of the intersection of the languages of two ED strings using techniques from formal languages and automata theory.

\begin{definition}[NFA]
\emph{A nondeterministic finite automaton} (NFA) is a 5-tuple $(Q,\Sigma,\delta,q_0,F)$, where $Q$ is a finite set of states; $\Sigma$ is an alphabet; $\delta:Q\times(\Sigma\cup\{\varepsilon\})\rightarrow \mathcal{P}(Q)$ is a transition function, where $\mathcal{P}(Q)$ is the power set of $Q$; $q_0\in Q$ is the starting state; and $F\subseteq Q$ is the set of accepting states.
\end{definition}

Using the folklore product automaton construction, one can check whether two NFA have a nonempty intersection in $\cO(N_1\cdot N_2)$ time, where $N_1$ and $N_2$ are the sizes of the two NFA~\cite{FiniteAutomata2004}.
We use a different, compacted representation of automata, which in some special cases allows a more efficient algorithm for computing and representing the intersection.

\begin{definition}[Compacted NFA]
An \emph{extended transition} is a transition function of the form $\delta^{\mathit{ext}}:Q\times\Sigma^*\rightarrow \mathcal{P}(Q)$, where $Q$ is a finite set of states, $\Sigma^*$ is the set of strings over alphabet $\Sigma$,
and $\mathcal{P}(Q)$ is the power set of $Q$.
A \emph{compacted NFA} is an NFA in which we allow extended transitions.
Such an NFA can also be represented by a standard (uncompacted) NFA, where each extended transition is subdivided into standard one-letter transitions (and $\varepsilon$-transitions), $\delta:Q\times(\Sigma\cup\{\varepsilon\})\rightarrow \mathcal{P}(Q)$. 
The states of the compacted NFA are called \emph{explicit}, while the states obtained due to these subdivisions are called \emph{implicit}.
\end{definition}

Given a compacted NFA $A$ with $V$ explicit states and $E$ extended transitions, we denote by $V^u$ and $E^u$ the number of states and transitions, respectively, of its uncompacted version $A^u$. Henceforth we assume that in the given NFA every state is reachable, and hence we have $V^u=\cO(E^u)$ and $V=\cO(E)$.

\begin{lemma}\label{lem:NFA-intersection}
Given two compacted NFA $A_1$ and $A_2$, with $V_1$ and $V_2$ explicit states and $E_1$ and $E_2$ extended transitions, respectively, a compacted NFA representing the intersection of $A_1$ and $A_2$ with $\cO(V_1^u V_2+V_1V_2^u)$ explicit states and $\cO(E_1^uE_2+E_1E_2^u)$ extended transitions can be computed in $\cO(E_1^uE_2+E_1E_2^u)$ time if $A_1$ and $A_2$ are over an integer alphabet $[1,(E_1^u+E_2^u)^{\cO(1)}]$. 
\end{lemma}
\begin{proof}
We start by constructing an LCP data structure over the concatenation of all the labels of extended transitions of both NFA of total size $\cO(E_1^u+E_2^u)$.
It requires $\cO(E_1^u+E_2^u)$-time preprocessing and allows answering LCP queries on any two substrings of such labels in $\cO(1)$ time.

We construct $B$, a compacted NFA representing the intersection of $A_1$ and $A_2$.

Every state of $B$ is composed of a pair: an explicit state of one automaton and any explicit or implicit state of the other automaton (or equivalently a state of the uncompacted version of the automaton). Thus the total number of explicit states of $B$ is $\cO(V_1^uV_2+V_1V_2^u)$.

We need to compute the extended transitions of $B$.
For a state $(u,v)$ we check every string pair $(P,Q)$, where $P$ iterates over all extended transitions going out of $u$ and $Q$ iterates over all extended transitions going out of $v$ (a transition going out of an implicit state is represented by a suffix of the transition it belongs to). For every pair $(P,Q)$ we ask an $\LCP(P,Q)$ query.
If $\LCP(P,Q)$ is equal to one of $|P|$, $|Q|$ (possibly both), we create an extended transition between $(u,v)$ and the pair of states reachable through those transitions (if one of the transitions is strictly longer, we prune it to the right length, ending it at an implicit state of its input NFA).
Otherwise such a transition does not lead to any explicit state of $B$ and thus cannot be used to reach the accepting state; hence we ignore it.

Finally, the starting (resp.~accepting) state of $B$ corresponds to a pair of starting (resp.~accepting) states of $A_1$ and $A_2$.

Since any pair representing an explicit state of $B$ contains an explicit state of $A_1$ or $A_2$, the number of such transition pair checks (and hence also the number of the extended transitions of $B$) is $\cO(E_1^uE_2+E_1E_2^u)$.
Since each such check takes $\cO(1)$ time, the construction complexity follows.
Note that NFA $B$ may contain unreachable states; such states can be removed afterwards in linear time. The algorithms' correctness follows from the observation that $B^u$ is in fact the standard intersection automaton of $A_1^u$ and $A_2^u$ with some states, that do not belong to any path between the starting and the accepting states, removed.
\end{proof}

\begin{figure}
    \centering
    \includegraphics[width=0.8\textwidth]{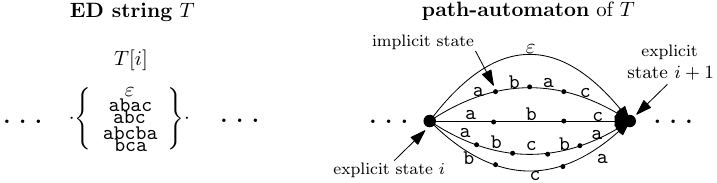}
    \caption{On the left: an ED string; on the right: the corresponding path-automaton.}
    \label{fig:automata}
\end{figure}

We next define the path-automaton of an ED string (inspect~\cref{fig:automata} for an example).

\begin{definition}[Path-automaton]
Let $T$ be an ED string of length $n$, cardinality $m$, and size $N$.
The \emph{path-automaton} of $T$ is the compacted NFA consisting of: 
\begin{itemize}
\item $V=n+1$ explicit states, numbered from $1$ through $n+1$.
State $1$ is the starting state and state $n+1$ is the accepting state. State $i\in[2,n]$ is the state \emph{in-between} $T[i-1]$ and $T[i]$.
\item  $m_i$ extended transitions from state $i$ to state $i+1$ labeled with the strings in $T[i]$, for all $i\in[1, n]$, where $E=m=\sum_i{m_i}$.
\end{itemize}
The \emph{path-automaton} of $T$ accepts exactly $\mathcal{L}(T)$.
The uncompacted version of this path-automaton has $V^u=\cO(N)$ states and $E^u=N$ transitions.

An example is constructed in Figure~\ref{fig:auto}.
\end{definition}

	\begin{figure}[ht]
		\centering
         \includegraphics[trim={0 13.5cm 0 0},clip,page=4,width=0.65\textwidth]{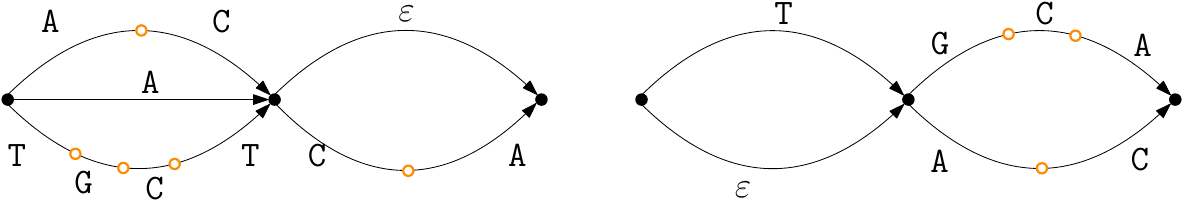}			\caption{The path-automata $A_1$ and $A_2$ for ED strings $T_1=     
\left\{\begin{array}{c}
          \texttt{AC} \\ \texttt{A}\\ \texttt{TGCT}
        \end{array}\right\}
        \cdot
         \left\{\begin{array}{c}
            \varepsilon\\ \texttt{CA}
        \end{array}\right\}$ and $T_2=
         \left\{\begin{array}{c}
             \texttt{T} \\ \varepsilon
        \end{array}\right\}
        \cdot
         \left\{\begin{array}{c}
            \texttt{GCA}\\ \texttt{AC}
        \end{array}\right\}$.
				The filled black nodes are explicit states, while the orange empty nodes are implicit states.
			}
			\label{fig:auto}
		\end{figure}

\cref{lem:NFA-intersection} thus implies the following result.

\begin{corollary}\label{cor:path-automata intersection}
The compacted NFA representing the intersection of two path-automata with $\cO(N_1n_2+N_2n_1)$ explicit states and $\cO(N_1m_2+N_2m_1)$ extended transitions can be constructed in $\cO(N_1m_2+N_2m_1)$ time if the path-automata are over an integer alphabet.
\end{corollary}

 	\begin{figure}[ht]
		\centering
		\includegraphics[page=2,width=0.8\textwidth]{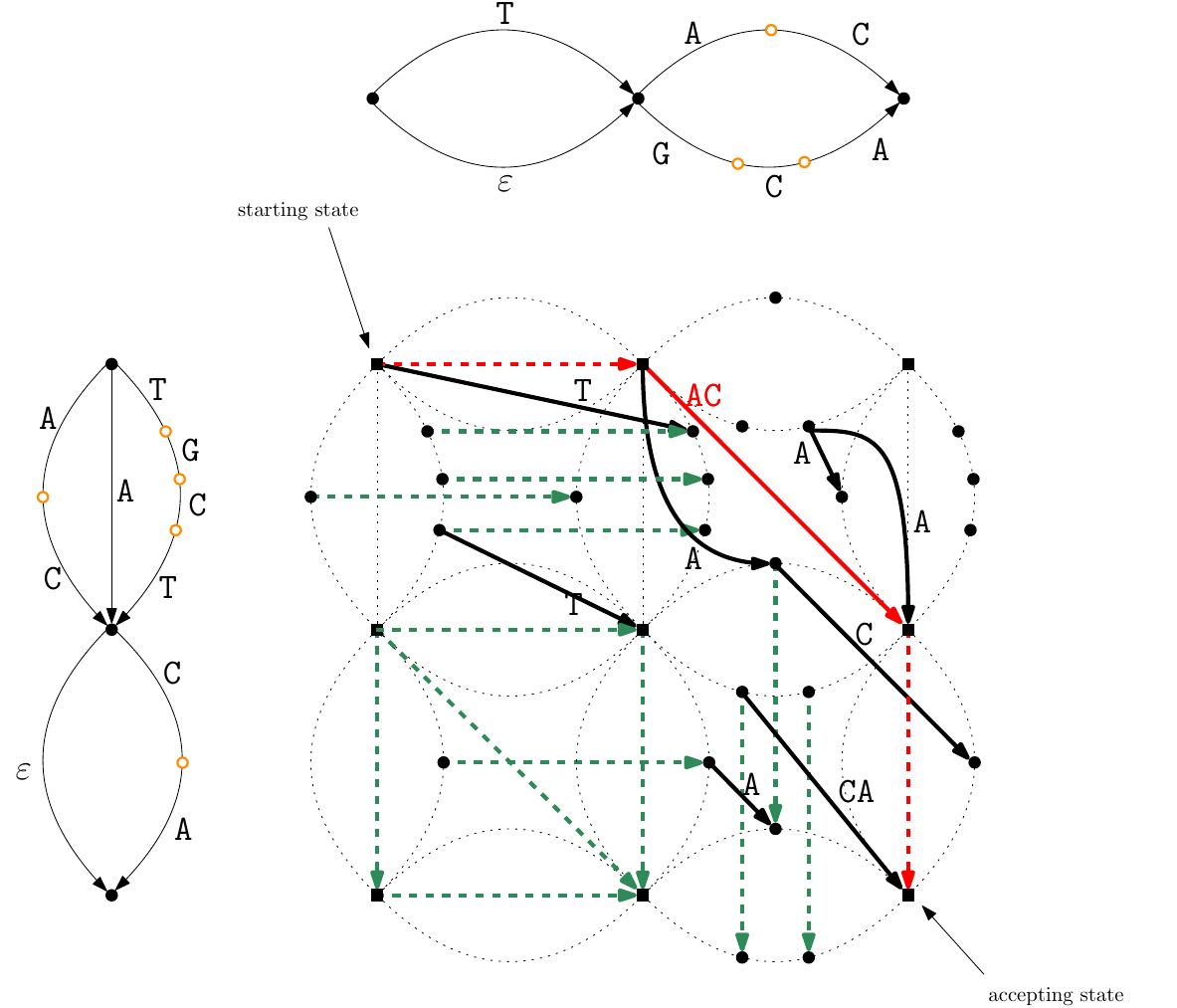}
		\caption{Intersection automaton for $T_1$ and $T_2$ as in Figure~\ref{fig:auto} where the string $\texttt{AC}$ in $\L(T_1) \cap \L(T_2)$ that determines a positive answer to the \EDSI can be spelled in the path from the starting state to the accepting state. The path-automata $A_1$ and $A_2$ are shown on the left and on the top, respectively, and nodes of the intersection automaton are arranged along dotted lines that correspond to copies of the layout of $A_1$ and $A_2$, to simplify the understanding of $G$. The dashed edges of the intersection automata correspond to $\varepsilon$-transitions (namely, transitions such that no letter is read when traversed), while the solid edges correspond to the other extended transitions.}
		\label{fig:edsi}
	\end{figure}

\begin{theorem}\label{thm:NFA-EDSI}
$\EDSI$ for ED strings over an integer alphabet can be solved in $\cO(N_1m_2+N_2m_1)$ time. If the answer is YES, a witness can be reported within the same time complexity.
\end{theorem}
\begin{proof}
The path-automaton of an ED string of size $N$ can be constructed in $\cO(N)$ time.
Given two ED strings, we can construct their path-automata in linear time and apply \cref{cor:path-automata intersection}.
By finding any path in the intersection automaton from the starting to the accepting state in linear time (if it exists), we obtain the result. The construction is detailed on an example in Figure~\ref{fig:edsi}
\end{proof}

Notice that the path-automata representing ED strings, as well as their intersection, are always acyclic, but may contain $\varepsilon$-transitions. In the following we are only interested in the graph underlying the path-automaton, that is the directed acyclic graph (DAG), where every \emph{node} represents an explicit state and every labeled directed \emph{edge} represents an extended transition of the path-automaton (inspect also~\cref{fig:NFA}). 

\subsection{An \texorpdfstring{$\ctO(N_1^{\omega-1}n_2+N_2^{\omega-1}n_1)$}{O(N1(omega-1) n2+ N2(omega-1) n1)}-time Algorithm for EDSI}\label{sec:MM}
In this section, we start by showing a construction of the \emph{intersection graph} computed by means of \cref{lem:NFA-intersection} in the case when the input is a pair of path-automata that allows an easier and more efficient implementation.
The construction is then adapted to obtain an $\ctO(N_1^{\omega-1}n_2+N_2^{\omega-1}n_1)$-time algorithm for solving the \EDSI problem.

For $x\in\{1,2\}$ by $A_x$ we denote the compacted NFA (henceforth, graph $A_x$) representing the ED string $T_x$. By $I_x[i]$ we denote the set of implicit states (henceforth, implicit nodes) appearing on the extended transitions (henceforth, edges) between explicit states (henceforth, explicit nodes) $i$ and $i+1$. For convenience, the implicit nodes in the sets $I_x[1],\ldots,I_x[n_x]$ can be numbered consecutively starting from $n_x+2$.

Let $U_{i,j}=\{(i,k): k \in \{j\}\cup I_2[j]\}$ and $U'_{i,j}=\{(k,j):k\in \{i\}\cup I_1[i]\}$, for all $i\in[1,n_1+1]$ and $j\in[1,n_2+1]$. As in the construction of~\cref{lem:NFA-intersection}, the union of all $U_{i,j}$ and $U'_{i,j}$ is the set of explicit nodes of the intersection graph that we construct; this can be represented graphically by a grid, where the horizontal and vertical lines correspond to $U_{i,j}$ and $U'_{i,j}$, respectively (inspect \cref{fig:grid}). In particular, we would like to compute the edges between these explicit nodes (inspect~\cref{fig:edges}) in $\cO(N_1m_2+N_2m_1)$ time.  

Consider an explicit node of the intersection graph; this node is represented by a pair of nodes: one from $A_1$ and one from $A_2$. We need to consider two cases: explicit \emph{vs} explicit node; or explicit \emph{vs} implicit node.
By symmetry, it suffices to consider an explicit  first node.
Let us denote this pair by $(i,k)\in U_{i,j}$, where $i$ is an explicit node of $A_1$ and $k$ is a node of $A_2$.
Let us further denote by $\ell_1$ the label of one of the edges going from node $i$ to node $i+1$.
For $k$, we have two cases. If $k$ is explicit (i.e., $k=j$) then we denote by $\ell_2$ the label of one of the edges going from $k$ to $k+1$. Otherwise ($k$ is implicit), we denote by $\ell_2$ the path label (concatenation of labels) from node $k$ to node $j+1$.

As noted in the proof of \cref{lem:NFA-intersection}, an edge is constructed only if $\textsf{LCP}(\ell_1,\ell_2)=\min(|\ell_1|,|\ell_2|)$. If $\textsf{LCP}(\ell_1,\ell_2)=|\ell_2|<|\ell_1|$ (a prefix of a string in $T_1[i]$ is equal to the suffix of a string in $T_2[j]$ starting at the position corresponding to node $k\in \{j\} \cup I_2[j]$), the edge ends in a node from $U'_{i,j+1}$
(\cref{fig:green}). If $\textsf{LCP}(\ell_1,\ell_2)=|\ell_1|<|\ell_2|$ (a whole string from $T_1[i]$ occurs in a string from $T_2[j]$ starting at the position corresponding to node $k\in \{j\} \cup I_2[j]$), the edge ends in a node from $U_{i+1,j}$
(\cref{fig:blue}).
Otherwise ($\textsf{LCP}(\ell_1,\ell_2)=|\ell_1|=|\ell_2|$; the two strings are equal) the edge ends in $(i+1,j+1)$.
Symmetrically (i.e., the second node is explicit), the edge going out of a node from $U'_{i,j}$ ends at a node from the set $U'_{i,j+1}\cup U_{i+1,j}\cup \{(i+1,j+1)\}$ (inspect~\cref{fig:edges}).

\begin{figure}[t]
     \centering
     \begin{subfigure}[b]{0.47\textwidth}
         \centering
         \includegraphics[width=0.75\textwidth]{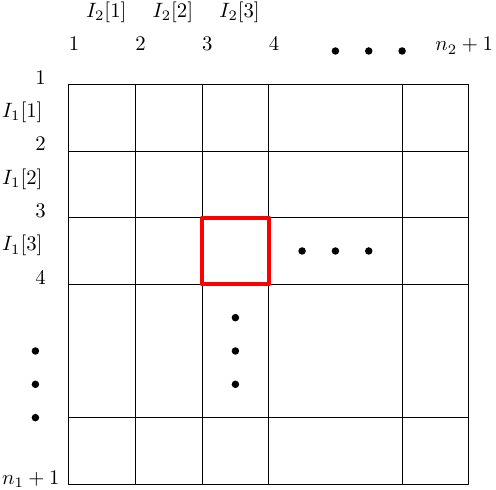}
         \caption{The grid of $U_{i,j}$ and $U'_{i,j}$.}
         \label{fig:grid}
     \end{subfigure}
     \hfill
     \begin{subfigure}[b]{0.47\textwidth}
         \centering
         \includegraphics[width=0.75\textwidth]{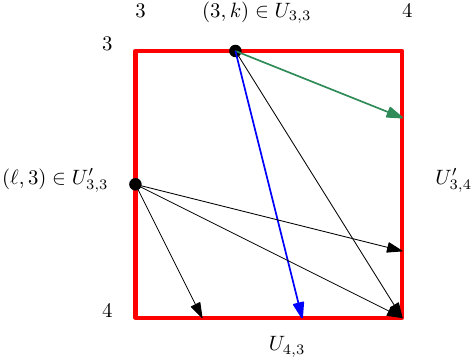}
         \caption{The edges.}
         \label{fig:edges}
     \end{subfigure}
     \hfill 
     \begin{subfigure}[b]{0.47\textwidth}
         \centering
         \includegraphics[width=0.75\textwidth]{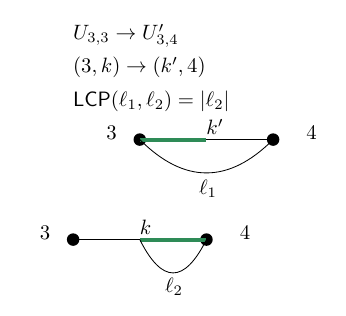}
         \caption{The green edge denotes a prefix-suffix match.}
         \label{fig:green}
     \end{subfigure}
     \hfill
     \begin{subfigure}[b]{0.47\textwidth}
         \centering
         \includegraphics[width=0.75\textwidth]{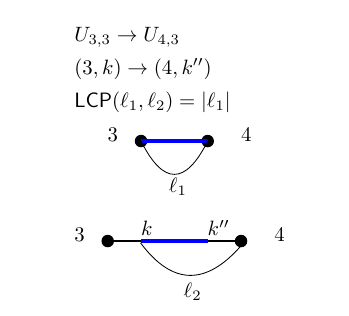}
         \caption{The blue edge denotes a full match.}
         \label{fig:blue}
     \end{subfigure}
        \caption{An overview of the edges computed by the algorithm.}
        \label{fig:grid-edges}
\end{figure}

We next show how to construct the intersection graph by computing all such edges going out of $U_{i,j}$ or $U'_{i,j}$ in a \emph{single batch} using suffix trees (inspect~\cref{fig:lem4} for an example). This construction allows an easier and more efficient implementation in comparison to the LCP data structure used in the general NFA intersection construction. Let us recall that $||T||$ denotes the size of the ED string $T$. Henceforth we denote $N_{1,i}=||T_1[i]||$ for $i \in [1,n_1]$ and $N_{2,j}=||T_2[j]||$ for $j \in [1,n_2]$.

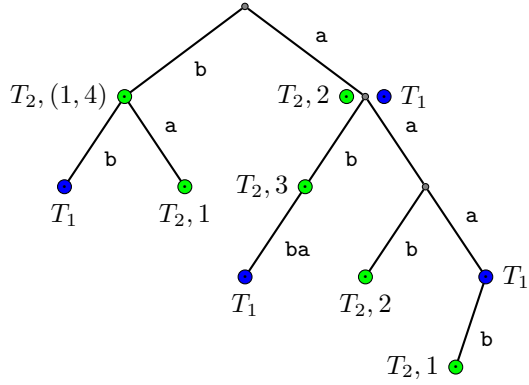
\begin{figure}[ht]
\centering
\begin{tikzpicture}[yscale=0.8,xscale=0.8,auto,node distance=0.1cm]
\usetikzlibrary{calc}

\tikzstyle{dot}=[inner sep=0.03cm, circle, draw]

\node[] at (11,11.7){};
\node[dot,fill=gray] (t) at (11,11) {};
\foreach \name/\dx/\parent/\l in {
  l/-2/t/\texttt{b},
  r/+2/t/\texttt{a},
  lr/+1/l/\texttt{a},
  ll/-1/l/\texttt{b},
  rl/-1/r/\texttt{b},
  rr/+1/r/\texttt{a},
  rll/-1/rl/\texttt{ba},
  rrl/-1/rr/\texttt{b},
  rrr/+1/rr/\texttt{a},
  rrrr/-0.5/rrr/\texttt{b}
  } {
  \node[dot,fill=gray] (\name) at ($(\parent)+(\dx,-1.5)$) {};
  \draw[black,thick] (\parent)--(\name) node[midway,auto] {\small \l};
}

\node[dot,fill=green,label={left:$T_2,(1,4)$}] (new) at (l) {.};
\node[dot,fill=green, label={below:$T_2,1$}] (new) at (lr) {.};
\node[dot,fill=green, label={left:$T_2,2$}] (new) [left= of r] {.};
\node[dot,fill=green, label={left:$T_2,3$}] (new) at (rl) {.};
\node[dot,fill=green, label={below:$T_2,2$}] (new) at (rrl) {.};
\node[dot,fill=green, label={left:$T_2,1$}] (new) at (rrrr) {.};
\node[dot,fill=blue,label={below:$T_1$}] (new) at (rll) {.};
\node[dot,fill=blue,label={right:$T_1$}] (new) at (rrr) {.};
\node[dot,fill=blue,label={below:$T_1$}] (new) at (ll) {.};
\node[dot,fill=blue,label={right:$T_1$}] (new) [right= of r] {.};

\end{tikzpicture}
\caption{The annotated compacted trie constructed for 
$T_1[i]=\begin{Bmatrix}
\texttt{abba}\\\texttt{aaa}\\\texttt{bb}\\\texttt{a}
\end{Bmatrix}$ and 
$T_2[j]=\begin{Bmatrix}
\texttt{ba}\\\texttt{aaab}\\\texttt{b}
\end{Bmatrix}$ in~\cref{lem:suffix tree graph construction}.
The node corresponding to $\texttt{b}$ has two $T_2$ labels and is an ancestor of the node corresponding to $\texttt{bb}$ with a $T_1$ label; hence two corresponding edges to $U'_{i,j+1}$ are constructed. The node corresponding to $\texttt{aaa}$ has a $T_1$ label and is an ancestor of the node corresponding to $\texttt{aaab}$ with a $T_2$ label; hence a corresponding edge to $U_{i+1,j}$ is constructed. The node corresponding to $\texttt{a}$ has both a $T_1$ and a $T_2$ label; hence a corresponding edge to $(i+1,j+1)$ is constructed.}\label{fig:lem4}
\end{figure}

\begin{lemma}\label{lem:suffix tree graph construction}
For any $i\in[1,n_1+1]$ and $j\in[1,n_2+1]$, we can construct all $K_{i,j}$ edges going out of nodes in $U_{i,j}$ in $\cO(N_{1,i}+N_{2,j}+K_{i,j})$ time using the generalized suffix tree of the strings in $T_2[j]$. We assume that the letters of strings in $T_2[j]$ are over an integer alphabet $[1,N_{2,j}^{\cO(1)}]$.
\end{lemma}
\begin{proof}
Let us start with a simple implementation of the lemma that works \emph{with high probability}. We will then give the details for a deterministic implementation.

We first construct the generalized suffix tree of the strings in $T_2[j]$ in $\cO(N_{2,j})$ time~\cite{DBLP:conf/focs/Farach97}.
We also mark each node corresponding to a suffix of a string in $T_2[j]$ with a $T_2$-label.
Each such node is also decorated with one or multiple starting positions, respectively, from one or multiple elements of $T_2[j]$ sharing the same suffix. For each branching node of the suffix tree, we construct a hash table, to ensure that any outgoing edge can be retrieved in constant time based on the first letter (the key) of its label.
This can be done in $\cO(N_{2,j})$ time with perfect hashing~\cite{DBLP:journals/jacm/FredmanKS84}. We next spell each string from $T_1[i]$ from the root of the suffix tree making implicit nodes explicit or adding new ones if necessary to create the compacted trie of all those strings; and, finally, we mark the reached nodes of the suffix tree with a $T_1$-label. Spelling all strings from $T_1[i]$ takes $\cO(N_{1,i})$ time.

Every pair of different labels marking two nodes in an ancestor-descendant relationship corresponds to exactly one outgoing edge of the nodes in $U_{i,j}$: (i) if a node marked with a $T_2$-label is an ancestor of a node marked with a $T_1$-label, then the suffix of a string from $T_2[j]$ matches a prefix of a string from $T_1[i]$ forming an edge ending in $U'_{i,j+1}$; (ii) if a node marked with a $T_1$-label is an ancestor of a node marked with a $T_2$-label, then a string from $T_1[i]$ occurs in a string from $T_2[j]$ extending its prefix and forming an edge ending in $U_{i+1,j}$; (iii) if a node is marked with a $T_1$-label and with a $T_2$-label, then the suffix of a string from $T_2[j]$ matches a string from $T_1[i]$ forming an edge ending in $(i+1,j+1)$. After constructing the generalized suffix tree of $T_2[j]$ and spelling the strings from $T_1[i]$, it suffices to make a DFS traversal on the annotated tree to output all $K_{i,j}$ such pairs of nodes. 

Let us note that the perfect hashing can be avoided under the assumption of the lemma stating that strings in $T_1[i] \cup T_2[j]$ are over an integer alphabet $[1,(N_{1,i}+N_{2,j})^{\cO(1)}]$). It suffices to construct the generalized suffix tree of $T_1[i] \cup T_2[j]$. Then one can trim all the nodes of the tree that do not have in their subtree a node corresponding to a (whole) string in $T_1[i]$ or a substring of a string in $T_2[j]$. This makes the construction deterministic.
\end{proof}

\begin{theorem}\label{the:ST}
We can construct the intersection graph of $T_1$ and $T_2$ in $\cO(N_1m_2+N_2m_1)$ time using the suffix tree data structure and tree search traversals if $T_1$ and $T_2$ are over an integer alphabet $[1,(N_1+N_2)^{\cO(1)}]$.
\end{theorem}
\begin{proof}
We will apply \cref{lem:suffix tree graph construction} for $U_{i,j}$ and $U'_{i,j}$, for all $i\in[1,n_1+1]$ and $j\in[1,n_2+1]$. To this end, before computing $U_{i,j}$ or $U'_{i,j}$, for each $i\in[1,n_1+1]$ and $j\in[1,n_2+1]$ we may need to renumber the letters in strings in $T_1[i]$ and $T_2[j]$ with consecutive integers to make sure that they belong to an integer alphabet $[1,(||T_1[i]||+||T_2[j]||)^{\cO(1)}]$. This can be done in $\cO(N_1n_2+n_1N_2)$ total time using one global radix sort.

We have that the total number of nodes is $\sum_{i,j} \cO(N_{1,i}+N_{2,j})=\cO(N_1n_2+N_2n_1)$, and then the number of all output edges is bounded by $\cO(N_1m_2+N_2m_1)$ by \cref{cor:path-automata intersection}. 
\end{proof}

Note that if we are interested only in checking whether the intersection is nonempty, and not in the computation of its graph representation, it suffices to check which of the nodes are \emph{reachable} from the starting node, which may be more efficient as there are $\cO(N_1n_2+N_2n_1)$ explicit nodes in this graph.

Let $X$ be the set of nodes of $U_{i,j}$ that are reachable from the starting node. From this set of nodes we need to compute two types of edges (inspect~\cref{fig:edges}). The first type of edges, namely, the ones from $X$ to $U'_{i,j+1}\cup\{(i+1,j+1)\}$ (green edges in~\cref{fig:edges}), are computed by means of~\cref{lem:pref=suf}, which is similar to~\cref{lem:suffix tree graph construction}.
For the second type of edges, namely, the ones from $X$ to $U_{i+1,j}\cup\{(i+1,j+1)\}$ (blue edges in~\cref{fig:edges}), we use a reduction to the so-called \emph{active prefixes extension} problem~\cite{elasticSICOMP} (\cref{lem:extensions}).

\begin{lemma}\label{lem:pref=suf}
For any given $X\subseteq U_{i,j}$, we can compute the subset of $U'_{i,j+1}\cup\{(i+1,j+1)\}$ containing all and only the nodes that are reachable from the nodes of $X$ in $\cO(N_{1,i}+N_{2,j})$ time. We assume that the letters of strings in $T_1[i] \cup T_2[j]$ are over an integer alphabet $[1,(N_{1,i}+N_{2,j})^{\cO(1)}]$.
\end{lemma}
\begin{proof}
In~\cref{lem:suffix tree graph construction}, the edges from nodes of $U_{i,j}$ to nodes of $U'_{i,j+1}$ come from a pair of nodes in the generalized suffix tree of $T_2[j]$ enriched with strings from $T_1[i]$: one marked with a $T_1$-label and its ancestor marked with a $T_2$-label. Notice that the $T_2$-labels are in a correspondence with the elements of $U_{i,j}$ (the labels on a proper suffix of a string in $T_2$ are in a one-to-one correspondence with $U_{i,j}\backslash \{(i,j)\}$, and $(i,j)$ corresponds to whole strings in $T_2[j]$), and hence we can trivially remove the $T_2$-labels that do not correspond to the elements of $X$. Furthermore, we are not interested in the set of starting positions decorating a node with a $T_2$-label; we are interested only in whether a node is $T_2$-labeled or not (i.e., we do not care from which node of $X$ the edge originates). Since the nodes marked with a $T_1$-label have in total $N_{1,i}$ ancestors (including duplicates), we can compute the result of this case in $\cO(N_{1,i}+N_{2,j})$ time in total. Finally, the node $(i+1,j+1)$ is reachable when a single node is marked with both a $T_1$-label and a $T_2$-label. This can be checked within the same time complexity.
\end{proof}

The remaining edges (blue edges in~\cref{fig:edges}) are dealt with via a reduction to the following problem.

\defproblem{Active Prefixes (AP)}{A string $P$ of length $m$, a bit vector $W$ of size $m$, and a set $\mathcal{S}$ of strings of total length $N$.}{A bit vector $V$ of size $m$ with $V[p]=1$ if and only if there exists $P'\in\mathcal{S}$ and $p'\in[1, m]$, such that $P[1\dd p'-1]\cdot P' = P[1 \dd p-1]$ and $W[p']=1$.}

Bernardini et al.~have shown the following result in~\cite{elasticSICOMP}, which relies on fast matrix multiplication (FMM).

\begin{lemma}[\cite{elasticSICOMP}]\label{lem:APE}
The AP problem can be solved in $\ctO(m^{\omega-1})+\cO(N)$ time, where $\omega$ is the matrix multiplication exponent.
\end{lemma}

\begin{lemma}\label{lem:extensions}
For any given $X\subseteq U_{i,j}$, we can compute the subset of $U_{i+1,j}$ containing all and only the nodes that are reachable from the nodes of $X$ in $\ctO(N_{1,i}+N_{2,j}^{\omega-1})$ time. 
\end{lemma}
\begin{proof}
The problems of computing the subset of $U_{i+1,j}$ reachable from $X$ and the AP problem can be reduced to one another in linear time.

For the forward reduction, let us set $\mathcal{S}=T_1[i]$ and $P=\prod_{S\in T_2[j]} \$ S $, where $\$$ is a letter outside of the alphabet of $T_1$. This means that we order the strings in $T_2[j]$ in an arbitrary but fixed way. For a single string $\$ S$ (where $S\in T_2[j]$), the positions from $S[1\dd |S|-1]$ correspond to the implicit nodes (along the path spelling $S$) of $I_2[j]$, while the position with $\$$ corresponds to the explicit node $j$ of $A_2$. Through this correspondence, we can construct two bit vectors $W$ and $V$, each of them of size $|P|$, and whose positions are in correspondence with $\{j\}\cup I_2[j]$ (note that this correspondence is not a bijection, as the explicit node $j$ have several preimages when $|T_2[j]|\ge 2$). As $U_{i,j}$ and $U_{i+1,j}$ are in a 1-to-1 correspondence with $\{j\}\cup I_2[j]$, we use the same correspondence to match positions between $W$ and $U_{i,j}$ and between $V$ and $U_{i+1,j}$.
Finally, we set $W[k]=1$ if and only if the corresponding node of $U_{i,j}$ belongs to $X$.
After solving AP, we have $V[k]=1$ for some $k$ corresponding to a node of $U_{i+1,j}$ if and only if this node is reachable from $X$. 

In more detail, observe that since $\$$ does not belong to the alphabet of $T_1$, a string $S$ from $T_1[i]$ has to match a fragment of a string from $T_2[j]$ to set $V[k]$ to $1$. This happens only if additionally $W[k-|S|]=1$; both things happen at the same time exactly when: (i) there exists a node $(i,\ell)\in X$; (ii) there exists an edge from $(i,\ell)$ to $(i+1,\ell')$; and (iii) the positions $k-|S|$ and $k$ in $P$ correspond to $\ell,\ell'$, respectively.

In the above reduction we have $|P|=\sum_{S\in T_2[j]}(|S|+1)=\cO(N_{2,j}),$ and $||\mathcal{S}||=N_{1,j}$, hence the lemma statement follows by \cref{lem:APE}.

For the reverse reduction, given an instance of AP, we encode it by setting $T_1[i]=\mathcal{S}$, $T_2[j]=\{P\}$ ($N_{1,i}=||\mathcal{S}||$, $N_{2,j}=|P|$) and $X$ containing the nodes corresponding to positions $k$ where $W[k]=1$.

This reduction shows that a more efficient solution to the problem of finding the endpoints of edges originating in $X$ would result in a more efficient solution to the AP problem.
\end{proof}

\begin{theorem}\label{thm:omega-1 intersection}
We can solve $\EDSI$ in $\ctO(N_1^{\omega-1}n_2+N_2^{\omega-1}n_1)$ time, where $\omega$ is the matrix multiplication exponent. If the answer is YES, we can output a witness within the same time complexity.
\end{theorem}
\begin{proof}
It suffices to set the starting node $(1,1)$ as reachable, apply \cref{lem:pref=suf,lem:extensions}, and their symmetric versions for $U'_{i,j}$, for each value of $(i,j)\in [1, n_1+1]\times[1, n_2+1]$ in lexicographical order, with $X$ equal to the set of reachable nodes of $U_{i,j}$ (respectively of $U'_{i,j}$); and, finally, check whether node $(n_1+1,n_2+1)$ is set as reachable. Before applying \cref{lem:pref=suf}, each time we renumber the alphabet in $T_1[i] \cup T_2[j]$ to make it integer.
We bound the total time complexity of the algorithm by:
\[\sum_{i,j}\ctO(N_{1,i}^{\omega-1}+N_{2,j}^{\omega-1})=\ctO(n_2\sum_{i}N_{1,i}^{\omega-1}+n_1\sum_{j}N_{2,j}^{\omega-1})=\ctO(N_1^{\omega-1}n_2+N_2^{\omega-1}n_1).\]

If $\mathcal{L}(T_1)\cap \mathcal{L}(T_2)$ is nonempty, that is, if the node $(n_1+1,n_2+1)$ is set as reachable from node $(1,1)$, then we can additionally output a witness of the intersection -- a single string from $\mathcal{L}(T_1)\cap \mathcal{L}(T_2)$ -- within the same time complexity.
To do that we mimic the algorithm on the graph with reversed edges. This time, however, we do not mark all of the reachable nodes; we rather choose a single one that was also reachable from $(1,1)$ in the forward direction.
This way, the marked nodes form a single path from $(1,1)$ to $(n_1+1,n_2+1)$. The witness is obtained by reading the labels on the edges of this path.
\end{proof}

Observe that if $n_2=1$, that is, $T_2$ is simply a set of standard strings, no node in $U'_{i,1}$ other than $(i,1)\in U_{i,1}$ is reachable. Indeed, nodes $(i,1)$ can be reached from $(1,1)$ through $\varepsilon$-transitions from $T_1$, but reaching other nodes would require reading a letter, that is, also a change in the state of $T_2$, and the only explicit state other than $1$ in $T_2$ is $2$ (and $(i,2)\in U_{i,2}$).
Due to this, the symmetric version of \cref{lem:extensions} never needs to be used to compute transitions between $U'_{i,1}$ and $U'_{i+1,2}$. This allows for a more efficient solution in this case.

The same observation improves the time complexity of the algorithms in \cref{thm:NFA-EDSI,the:ST} in the case that $n_2=1$. It suffices to consider edges outgoing from nodes $(x,y)$ in the intersection graph such that $x$ is explicit in $A_1$; the number of such edges is $\cO(N_2m_1)$.

\begin{corollary}\label{cor:one-standard}
 If $n_2=1$ then the running time in \cref{thm:NFA-EDSI,the:ST} is $\cO(N_1+N_2m_1)$ and the running time in
 \cref{thm:omega-1 intersection} is $\ctO(N_1+N_2^{\omega -1}n_1)$.
\end{corollary}

This observation will be useful in case of the generalized versions of \EDSI showed in \cref{sec:DEDSM,sec:AEDSI}, where we can compare the running time of our algorithms with the running time of the existing solutions solving special cases of those \EDSI generalizations.

\section{Acronym Generation}\label{sec:AG}

In this section, we study a problem on standard strings. Given a sequence $P = P_1,\ldots,P_n$ of $n$ strings we define an \emph{acronym} of $P$ as a string $A=A_1\cdots A_n$, where $A_i$ is a (possibly empty) prefix of $P_i$, $i\in[1,n]$. We next formalize the \textsc{Acronym Generation} problem.

\defproblem{Acronym Generation (AG)}{A set $D$ of $k$ strings of total length $K$ and a sequence $P = P_1,\ldots,P_n$ of $n$ strings of total length $N$.}{YES if some acronym of $P$ is an element of $D$, NO otherwise.}

The AG problem is underlying real-world information systems (e.g., see \url{https://acronymify.com/} or \url{https://acronym-generator.com/}) and existing approaches rely on brute-force algorithms or heuristics to address different variants of the problem~\cite{DBLP:journals/amai/JacobsIW20,DBLP:conf/acl-louhi/KirchhoffT16,DBLP:conf/aaai/KubalN21,DBLP:journals/bmcbi/KuoLLH09,DBLP:journals/isci/LiuLH17,DBLP:conf/psb/SchwartzH03,DBLP:journals/ijdar/TaghvaG99,DBLP:conf/coling/VeysehDTN20}. These algorithms usually accept a sequence $P$ of $n\leq n_{\max}$ strings, for some small integer $n_{\max}$, which highlights the lack of efficient exact algorithms for generating acronyms. Here we show an exact polynomial-time algorithm to solve AG for any $n$.

We can encode AG by means of \EDSI and modify the developed methods. 
The AG problem reduces to \EDSI in which $T_1[i]$, $i\in[1,n]$, is the set of all prefixes of $P_i$ and $T_2[1]=D$.
We have $N_1 \le N^2$, $m_1=N$, $n_1=n$, $N_2=K$, $m_2=k$, and $n_2=1$, so by \cref{cor:one-standard}, we obtain solutions to the AG problem working in $\cO(N^2+KN)$ and $\ctO(N^2+K^{\omega -1}n)$ time, respectively.

Since, however, all elements of set $T_1[i]$ are prefixes of a single string ($P_i$), we can obtain a more efficient graph representation of $T_1$ by joining nodes $i$ and $i+1$ with a single path labeled with $P_i$, with an additional $\varepsilon$ edge between every (implicit) node of the path and node $i+1$. As the size of the graph for $T_1$ is smaller ($\cO(N)$ nodes and edges), by using \cref{lem:NFA-intersection} we obtain an $\cO(kN+KN)=\cO(KN)$-time algorithm for solving the AG problem.

The considered ED strings have additional strong properties. $T_1[i]$'s are not just sets of prefixes of single strings, but sets of all their prefixes, while the length $n_2$ of $T_2$ is equal to $1$.
By employing these two properties we obtain the following improved result.
\begin{figure}
    \centering
    \includegraphics[width=0.45\textwidth]{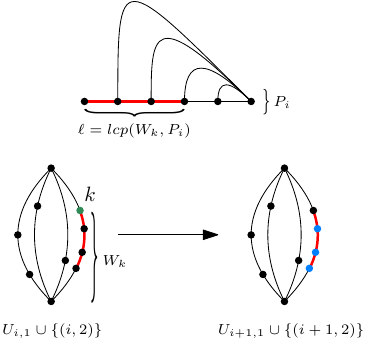}
    \caption{For the AG problem, the automaton for the prefixes of $P_i$  can be represented as the one on top of the figure. The sets $U_{i,1}\cup\{(i,2)\}$ and $U_{i+1,1}\cup\{(i+1,2)\}$ can be represented as copies of the path-automaton of $D$, as in the general case. A green node $k$ in $U_{i,1}\cup\{(i,2)\}$ is specified. The red segments form a path reading the LCP string between $W_k$ and $P_i$. The nodes that are reachable from $k$ in $U_{i+1,1}\cup\{(i+1,2)\}$ are highlighted in blue.}
    \label{fig:AG extension}
\end{figure}
\begin{theorem}\label{thm:AG}
AG can be solved in $\cO(nK+N)$ time.
\end{theorem}
\begin{proof}
The algorithm of \cref{thm:omega-1 intersection} is based on finding out which elements of sets $U_{i,j}, U'_{i,j}$ are reachable. However, since $n_2=1$, the sets $U'_{i,j}$ are trivialized. As in \cref{cor:one-standard}, it is enough to focus on the transitions between $U_{i,j}$ and $U_{i+1,j}\cup \{(i+1,2)\}$.

In \cref{lem:extensions}, to compute the reachable nodes of $U_{i+1,j}$ knowing the reachable nodes of $U_{i,j}$, fast matrix multiplication is employed (\cref{lem:APE}), but in this special case a simpler method will be more effective. 
For $k \ne 1$, let $W_k$ be the unique string read between nodes $k$ and $2$ in the path-graph of $T_2$. The crucial observation is: the edges going out of node $(i,k)\in U_{i,1}\cup\{(i,2)\}$ for $k\neq 1$ end in nodes $(i+1,k')$ for $k'\in[k, k+\ell]$, where $\ell=\LCP(P_i,W_k)$ as the strings from $T_1[i]$ matching the prefix of $W_k$ are exactly all the prefixes of $P_i$ of length at most $\ell$ (see Figure~\ref{fig:AG extension}).

Hence, to compute the reachable subset of $U_{i+1,1}\cup\{(i+1,2)\}$, we can handle the edges going out of $(i,1)$ separately in $\cO(K+|P_i|)$ time by letter comparisons, then compute the $\LCP(P_i,W_k)$ values for all the reachable nodes $(i,k)$ either using the LCP data structure, or with the use of the generalized suffix tree of $T_2[1]=D$ in $\cO(K+|P_i|)$ total time, and finally, using a 1D line sweep approach, compute the union of the obtained intervals in $\cO(K)$ time. In the end, we answer YES if and only if node $(n+1,2)$ is reachable.

Over all values of $i$ this gives an algorithm running in $\sum_{i}\cO(K+|P_i|)=\cO(nK+N)$ total time.
\end{proof}

Furthermore, one may be allowed to choose, for each $i\in[1,n]$, a lower bound $x_i \ge 0$ on the length of the prefix of $P_i$ used in the acronym (some strings should not be completely excluded from the acronym). The only modification to the algorithm in such a generalization is replacing intervals $[k, k+\ell]$ by $[k+x_i, k+\ell]$, which does not influence the claimed complexity. 

\begin{corollary}
If the answer to the instance of the AG problem is YES, we can output all strings in $D$ which are acronyms of $P$ within $\cO(nK+N)$ time.
\end{corollary}
\begin{proof}
We want to find out which strings of $T_2$ represent acronyms of $P$.
In the algorithm employed by \cref{thm:AG} the reachable nodes of $U_{i,1}\cup\{(i+1,2)\}$ are found.
If there exists an edge from a reachable node $(i,k)$ for $k\notin \{1,2\}$ to $(i+1,2)$ for some $i\in[1,n_1]$, then the path of the path-graph of $T_2$ containing node $k$ represents an acronym of $P$.
If node $(i+1,2)$ is reached directly from reachable node $(i,1)$, then the whole prefix of $P_i$ used to do that is in $D$, and hence is a standalone acronym of $P$.
If for a path neither of the two cases qualifies, then it cannot be used to reach node $(n+1,2)$, and hence is not an acronym of $P$.
\end{proof}

If the generalization with minimal lengths of prefixes is applied, then the values of $i$ used here are restricted to $[i', n]$, where $i'$ is the largest value of $i$ with a restriction $x_{i}>0$: node $(i'-1,2)$ does not have an edge to node $(i',2)$, and hence does not belong to any path from $(1,1)$ to $(n+1,2$). 

Let us remark that although the main focus of real-world acronym generation systems is on the natural language parsing and interpretation of acronyms, our new algorithmic solution may inspire practical improvements in such systems or further algorithmic work.

\section{ED String Comparison Tasks}\label{sec:EDapps}
In this section, we show some applications of our techniques from~\cref{{sec:EDSI}}. In particular, we show how intersection graphs can be used to solve different ED string comparison tasks.

We consider two ED strings, $T_1$ of length $n_1$, cardinality $m_1$ and size $N_1$, and $T_2$ of length $n_2$, cardinality $m_2$ and size $N_2$. We call \emph{intersection graph} the underlying graph of an automaton representing $\mathcal{L}(T_1)\cap \mathcal{L}(T_2)$. By Corollary~\ref{cor:path-automata intersection} such an automaton (and therefore the corresponding intersection graph) can be constructed in $\cO(N_1m_2+N_2m_1)$ time. In \cref{sec:EDSI}, such a graph was used to check whether $\mathcal{L}(T_1)\cap \mathcal{L}(T_2)$ is nonempty. 

In the following, we present other applications of intersection graphs (computed by means of Corollary~\ref{cor:path-automata intersection} or \cref{lem:suffix tree graph construction}) to tackle several natural ED string comparison tasks with no additional time complexity. 
We always assume that $T_1$ and $T_2$ are over an integer alphabet $[1,(N_1+N_2)^{\cO(1)}]$.

\subsection{Shortest/Longest Witness}
Let us start with the most basic application.

\defproblem{EDSI Shortest/Longest Witness}{Two ED strings, $T_1$ of length $n_1$, cardinality $m_1$ and size $N_1$, and $T_2$ of length $n_2$, cardinality $m_2$ and size $N_2$.}{A shortest (resp.~longest) element of $\mathcal{L}(T_1)\cap \mathcal{L}(T_2)$ if it is a nonempty set, FAIL otherwise.}

\begin{fact}\label{fct:witness}
The \textsc{EDSI Shortest/Longest Witness} problem can be solved in $\cO(N_1m_2+N_2m_1)$ time by using an intersection graph of $T_1$ and $T_2$.
\end{fact}

\begin{proof}
We compute the intersection graph $G$ as the underlying graph of the automaton computed in Corollary~\ref{cor:path-automata intersection}. Given an edge $(k,k')$ in $G$, we assign a weight $w(k,k')$ equal to the length of its string label (the string $\varepsilon$ has length $0$). Note that, by construction, $G$ is a directed acyclic graph (DAG). Thus the problem reduces to computing the shortest or the longest path between a source and a sink of a DAG, a problem with a well-known linear-time solution that involves topological sorting.
By reading the labels on the shortest (resp.~longest) path in $\cO(N_1+N_2)$ time, we can output the shortest (resp.~longest) element of $\mathcal{L}(T_1)\cap\mathcal{L}(T_2)$.\footnote{In the case of a shortest witness, we can obtain an equally efficient algorithm by employing Dijkstra's algorithm using bucket queue since the bound $\cO(N_1+N_2)$ on the weight of the path is known beforehand.}
\end{proof}

\subsection{Counting Pairs of Matching Strings}

In the next task, we would like to compute the total number of matching pairs of strings in $\mathcal{L}(T_1)\times \mathcal{L}(T_2)$ considering multiplicities in $\mathcal{L}(T_1)$ and $\mathcal{L}(T_2)$. We assume that the multiplicity of a string $S$ in $\mathcal{L}(T_1)$ is the number of sequences $S_1\in T_1[1],\ldots, S_{n_1}\in T_1[n_1]$ such that $S_1\cdots S_{n_1}=S$. The definition for $T_2$ is analogous. 

\defproblem{EDSI Matching Pairs Count}{Two ED strings, $T_1$ of length $n_1$, cardinality $m_1$ and size $N_1$, and $T_2$ of length $n_2$, cardinality $m_2$ and size $N_2$.}{$|\{(S_1,S_2) \in \mathcal{L}(T_1)\cap\mathcal{L}(T_2)\,:\,S_1=S_2\}|$.}

Each matching pair can be represented by a pair of \emph{alignments}: the sequence of $n_1+n_2$ choices of a single production $S_{i}\in T_1[i]$ and $S'_j\in T_2[j]$ for each $i \in[1,n_1]$ and $j\in[1,n_2]$, such that the resulting standard string is the same. Such a pair of alignments can in turn be represented by a path in the intersection graph of $T_1$ and $T_2$ as the edges correspond to (parts of) productions in some $T_1[i]$ and in some $T_2[j]$.

The representation is almost unique; the only reason this does not need to be the case is when a node $(i+1,j+1)$ is reached from a node $(i,j)$ for $i\in[1,n_1], j\in [1,n_2]$ using only $\varepsilon$-edges.
In this case, the three subpaths $(i,j)\rightarrow (i+1,j+1), (i,j)\rightarrow (i,j+1)\rightarrow (i+1,j+1)$ and $(i,j)\rightarrow (i+1,j) \rightarrow (i+1,j+1)$, all correspond to the choice of productions $\varepsilon\in T_1[i]$ and $\varepsilon\in T_2[j]$, even though this choice should be counted as one.

To fix the notation, we call an $\varepsilon$-edge: \emph{vertical}, when it leads from a node $(i,k)$ to $(i+1,k)$ for some explicit node $i$ of $A_1$; \emph{horizontal}, when it leads from a node $(k,j)$ to $(k,j+1)$ for some explicit node $j$ of $A_2$; and \emph{diagonal} if it leads from a node $(i,j)$ to $(i+1,j+1)$, where both $i$ and $j$ are explicit nodes.

The problem becomes even more complicated when a few such $\varepsilon$-productions are used in a row in both ED strings (a node $(i+x,j+y)$ is reached from $(i,j)$ using only $\varepsilon$-edges), as a single alignment would correspond to all the ``down, right or diagonal'' paths in the $x\times y$ grid.
The number of such paths can be large, and even if we remove all such diagonal $\varepsilon$-edges (those can be always simulated with a single horizontal and a single vertical one), we cannot remove the horizontal or vertical ones, as those can be traversed by other paths independently. We are still left with ${x+y \choose x}$ equivalent subpaths from $(i,j)$ to $(i+x,j+y)$. In order to mitigate this problem we will restrict the usage of such subpaths of many $\varepsilon$-edges to a single, regular one.

We call a path \emph{$\varepsilon$-regular}, if it does not use diagonal $\varepsilon$-edges, and if two $\varepsilon$-edges are used sequentially, then they either have the same direction, or the second one is horizontal.

\begin{lemma}\label{lem:path correspondence}
There is a one-to-one correspondence between 
pairs of alignments that produce the same string and
$\varepsilon$-regular paths from the starting to the accepting node in the intersection graph of $T_1$ and $T_2$.
\end{lemma}
\begin{proof}
Consider a pair of alignments; it can be represented by the produced string together with some positions in-between letters marked with $(n_1+1)$ $T_1$-labels and $(n_2+1)$ $T_2$-labels in total specifying the beginning and ending of the productions used in $T_1[i]$ and $T_2[j]$, for $i\in[1,n_1],j\in[1,n_2]$. A single position can be marked with both types of label and even many labels of the same type (when $\varepsilon$-productions are used). Each position corresponds to a pair of states in the path-automata: those pairs of states with at least one label read from left to right form a path in the intersection graph, as the $T_i$-label shows that the state is explicit in $T_i$. If two labels of a different type appear in the same place, this corresponds to a node composed of two explicit states. When a single position contains multiple labels from both types,
the exact ordering between those labels represents the actual path in the graph. At the same time from the point of view of $T_1$ and $T_2$ separately all those orderings correspond to exactly the same pair of alignments. 
This is where $\varepsilon$-regularity plays its role and only one subpath is created (all $T_1$-labels are placed before the $T_2$-labels). This way we have defined an injection from the pairs of alignments to the paths of the intersection graph (each pair of alignments corresponds to only one $\varepsilon$-regular path).

On the other hand the labels of the edges in the path represent (the parts of) the transitions used in each automaton, and hence also a pair of alignments, thus the function is also surjective.

We have thus shown that the function relating a pair of alignments to a path is both injective and surjective, and hence a bijection (the one-to-one correspondence from the statement).
\end{proof}

The main result then follows. 

\begin{fact}
The \textsc{EDSI Matching Pairs Count} problem can be solved in $\cO(N_1m_2+N_2m_1)$ time by using an intersection graph of $T_1$ and $T_2$.
\end{fact}
\begin{proof}
We compute the intersection graph as the underlying graph of the automaton computed in \cref{cor:path-automata intersection} or \cref{lem:suffix tree graph construction}. By \cref{lem:path correspondence}, we are counting distinct $\varepsilon$-regular paths from the starting node to the accepting one. As previously, we can sort the nodes topologically.
If we wanted to count all the paths it would be enough to compute for each node $k$ the value $S(k)$ equal to the number of paths from the starting node to $k$. One would have $S(k)=\sum_{\{(k',k)\}\in E(G)} S(k')$, over all edges $(k',k)$ (including parallel edges with different labels) and setting $S((1,1))=1$ where $(1,1)$ is the starting node. 
This time, however, we apply slightly more complicated formulas, which count the paths ending with the $\varepsilon$-edges separately from the other ones, and among those separately depending on the direction of the $\varepsilon$-edge. For a node $k$, let $S(k)$, $S^h(k)$ and $S^v(k)$ denote the number of $\varepsilon$-regular paths from the starting node to the node $k$ that end with an edge with a non-empty label, with a horizontal $\varepsilon$-edge and with a vertical $\varepsilon$-edge, respectively.
\begin{itemize}
\item $S(k)=\sum_{\{(k',k) \text{ positive length edge}\}} S(k')+S^h(k')+S^v(k')$
\item $S^h(k)=\sum_{\{(k',k) \text{ horizontal $\varepsilon$-edge}\}} S(k')+S^h(k')+S^v(k')$
\item $S^v(k)=\sum_{\{(k',k) \text{ vertical $\varepsilon$-edge}\}} S(k')+S^v(k')$
\end{itemize}
By induction, for the accepting node $(n_1+1,n_2+1)$, the number of distinct $\varepsilon$-regular paths from the starting to the accepting node is equal to the sum of the values $S$, $S^h$ and $S^v$ computed for node $(n_1+1,n_2+1)$.
\end{proof}

\subsection{ED Matching Statistics}
Asking whether $\mathcal{L}(T_1)\cap\mathcal{L}(T_2)$ is not nonempty, as a way to tell if two ED strings have something in common, can be too restrictive in practical applications. We will thus consider two more elaborate ED string comparison tasks that consider local matches rather than a match that necessarily involves the entire ED strings from beginning to end. Both notions that we consider, \textsc{Matching Statistics} and \textsc{Longest Common Substring}, are heavily employed on standard strings for practical applications, especially in bioinformatics.

We start by extending the classic \textsc{Matching Statistics} problem~\cite{DBLP:books/cu/Gusfield1997} from the standard string setting to the ED string setting.
Although this solution has already been described in~\cite{frontiers}, we provide it also here as an intermediate step of the solution to the next comparison task.

\defproblem{ED Matching Statistics}{Two ED strings, $T_1$ of length $n_1$, cardinality $m_1$ and size $N_1$, and $T_2$ of length $n_2$, cardinality $m_2$ and size $N_2$.}{For each $i\in[1, n_1]$, the length $\textsf{MS}[i]$ of the longest prefix of a string in $\mathcal{L}(T_1[i\dd n_1])$ that is a substring of a string in $\mathcal{L}(T_2)$.}

\begin{fact}\label{fct:matching statistics}
The \textsc{ED Matching Statistics} problem can be solved in $\cO(N_1m_2+N_2m_1)$ time by using an intersection graph of $T_1$ and $T_2$.
\end{fact}

\begin{proof}
This time we will use a slightly augmented version of the intersection graph coming from the unpruned intersection automaton. Namely, we construct the automaton as in Corollary~\ref{cor:path-automata intersection}, but do not remove the unreachable parts or the ``partial transitions''. That is, when we process an explicit state corresponding to a pair $(u,v)$ of states in the path-automata $A_1$ and $A_2$, and a pair $(s,t)$ of transitions going out of $u$ and $v$, we construct the corresponding transition to the pair of states that can be reached through $s$ and $t$, even if this transition finishes in a pair of implicit states, namely when $0<\LCP(\ell_1,\ell_2)<\min(|\ell_1|,|\ell_2|)$, where $\ell_1$ and $\ell_2$ are the respective labels of $s$ and $t$. Even in that case, the number of transition pair checks remains the same, and therefore the total size of the constructed underlying graph $G$ stays $\cO(N_1m_2+N_2m_1)$.

Once again we assign to each edge the weight $w$ storing the length of its string label and process the nodes in the reversed topological order to compute for each node $k$ the value $M(k)$ equal to the length of the longest path from $k$ in $A_1$ that matches a path in $A_2$; we have $M(k)=\max_{k'} M(k')+w(k,k')$ where $k'$ iterates over all successors of $k$. One has $M(k)=0$ for the nodes that do not have successors (for example the accepting node or nodes corresponding to a pair of implicit states).

By construction, we have $M((i,v))=\ell$, for an explicit state $i$ of $A_1$ and a state $v$ of $A_2$, if and only if $\ell$ is equal to the maximal LCP between a pair $(S_1,S_2)$, where $S_1$ a string in $\mathcal{L}(T_1[i\dd n])$ and $S_2$ is a string read starting at (explicit or implicit) state $v$ in $A_2$. 

For every explicit state $i$ of $A_1$ we can compute $\textsf{MS}[i]= \max_v M((i,v))$ over all (explicit or implicit) states $v$ of $A_2$ to obtain the output.
\end{proof}

An example of this construction is shown in Figure~\ref{fig:G4ms}.

	\begin{figure}
		\centering
		\includegraphics[width=0.8\textwidth,page=5]{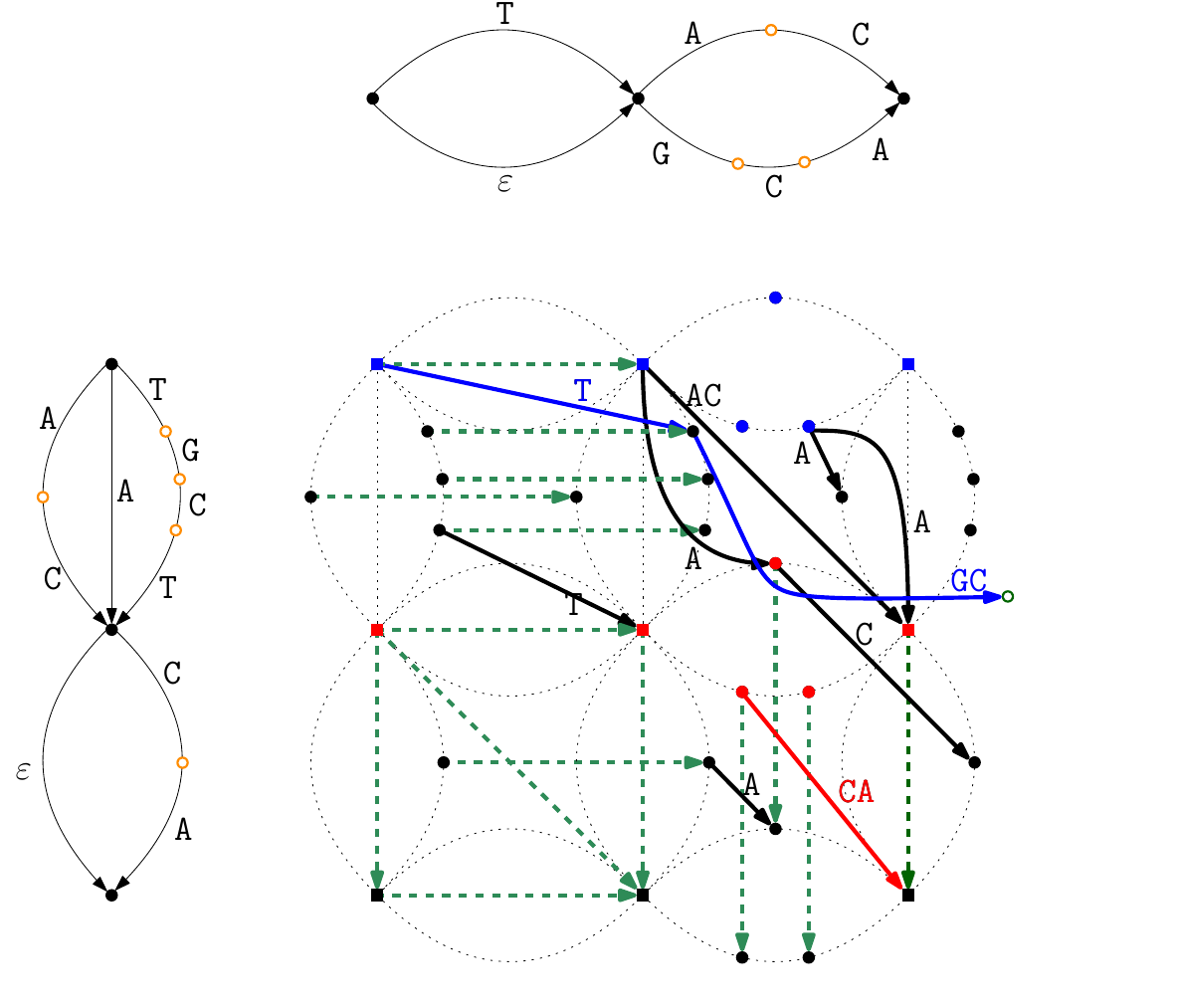}
		\caption{Matching statistics on the intersection graph of $T_1$ and $T_2$ as defined in Figure~\ref{fig:auto}. To simplify the understanding, we also draw $A_1$ and $A_2$ on the left and on the top, respectively. Note that this time, the pairs of implicit nodes that are reachable in a single extended transition from one pair that was previously computed are added. In the figure, there is only one such extra node that we represent by a green empty circle at the right of the graph. Here we highlight paths that are relevant for computing the Matching Statistics. In order to compute $\textsf{MS}[1]$, we look at the paths starting at nodes $(i,j)$ where $i$ is the explicit state $1$ in the path-automaton of $T_1$, and return the length of the longest label of such a path. These are the paths starting in one of the blue nodes (these are the nodes that correspond to the uppermost explicit node of $A_1$ paired with any node of $A_2$, that is, they correspond to the uppermost dotted copy of $A_2$); the longest of such paths (also drawn in blue) corresponds to the string $\texttt{TGC}$ having length $3$, and therefore, $\textsf{MS}[1]=3$. 
	    For $\textsf{MS}[2]$ we do the same, but using as starting nodes those in red that correspond to the internal explicit node of $A_1$ paired with any node of $A_2$ (that is, the nodes of the middle dotted copy of $A_2$). Here the longest path is drawn in red and it spells the string $\texttt{CA}$, and therefore we set $\textsf{MS}[2]=2$.}
		\label{fig:G4ms}
	\end{figure}

\subsection{ED Longest Common Substring}
We now proceed to extending the classic \textsc{Longest Common Substring} problem~\cite{DBLP:books/cu/Gusfield1997} from the standard string setting to the ED string setting.

\defproblem{ED Longest Common Substring}{Two ED strings, $T_1$ of length $n_1$, cardinality $m_1$ and size $N_1$, and $T_2$ of length $n_2$, cardinality $m_2$ and size $N_2$.}{A longest string that occurs in a string of $\mathcal{L}(T_1)$ and a string of $\mathcal{L}(T_2)$}

In the \textsc{ED Matching Statistics} problem only strings starting in explicit nodes of $A_1$ were considered; here we want to lift this restriction.
Computing the value of $M$ for every node of $A_1$ against every node of $A_2$ would take $\Theta(N_1N_2)$ time. We are only interested in the globally maximal value of $M$. Hence, we can focus only on the computation of those values for a certain subset of those pairs.

\begin{fact}
The \textsc{ED Longest Common Substring} problem can be solved in $\cO(N_1m_2+N_2m_1)$ time by using an intersection graph of $T_1$ and $T_2$.
\end{fact}
\begin{proof}

We start from the augmented graph $G$, as defined in the proof of \cref{fct:matching statistics}, and computed in $\cO(N_1m_2+N_2m_1)$ time. In addition to the standard edges between the explicit nodes, the graph contains edges from a pair containing at least one explicit state to a pair of implicit states, that are inclusion-wise maximal, that is, cannot be extended to obtain a longer edge. For the \textsc{ED Longest Common Substring} problem, we additionally need edges symmetric to those -- the (inclusion-wise maximal) edges starting in pairs of implicit states that end in pairs containing at least one explicit state.
    By a symmetric argument (argument for the reversed automata), there are $\cO(N_1m_2+N_2m_1)$ such edges and all can be computed together with the rest of the graph within the same time complexity. We compute the values $M(k)$ for each node $k$ as previously, and find their global maximum.
    Notice that if $M(k)$ for an implicit node $k$ associates $k$ with another implicit node $k'$, we do not need to compute the value $M(k)$. Indeed, then $(k,k')$ lies in the middle of an edge and the first node of its predecessor on the edge always has a greater value of $M$.
    
    We are left with the nodes that are not (weakly) connected to any explicit node. The isolated edges containing such nodes were not computed at all.
    Those edges correspond to strings that, in both ED strings, are fully contained in a single set. Hence, it is enough to compute the longest common substring of two strings, one being a concatenation of strings in $\bigcup_i T_1[i]\#$ and the other a concatenation of strings in $\bigcup_j T_2[j]\$$, for sentinel letters \#, \$. This can be done in $\cO(N_1+N_2)$ time.
    
    The method of obtaining the witness is the same as in the previous problems (after computing both endpoints of the optimal path using the algorithm described above).
\end{proof}

\subsection{ED Longest Common Subsequence}
Finally, we show how to extend the classic \textsc{Longest Common Subsequence} (\textsc{LCS}) problem~\cite{DBLP:books/cu/Gusfield1997} from the standard string setting to the ED string setting. We remark that, unlike the previous problems, in the standard string setting, \textsc{LCS} is not solvable in linear time: there exists a conditional lower bound saying that no algorithm with running time $\cO(n^{2-\epsilon})$, for any $\epsilon>0$, can solve \textsc{LCS} on two length-$n$ strings unless SETH fails~\cite{bringmann2015quadratic}.

\defproblem{ED Longest Common Subsequence}{Two ED strings, $T_1$ of size $N_1$, and $T_2$ of size $N_2$.}{A longest string that occurs in a string of $\mathcal{L}(T_1)$ and a string of $\mathcal{L}(T_2)$ as a subsequence.}

\begin{fact}
The \textsc{ED Longest Common Subsequence} problem can be solved in $\cO(N_1\cdot N_2)$ time by using an uncompacted intersection graph of $T_1$ and $T_2$.
\end{fact}
\begin{proof}
Consider the uncompacted path-automata for $T_1$ and $T_2$, having respective sizes $\cO(N_1)$, $\cO(N_2)$.
For every single transition in those graphs (that is a single letter transition since they are uncompacted), we add a parallel $\varepsilon$-transition (that is, we allow to skip the letter). The sizes of the automata remain $\cO(N_1)$ and $\cO(N_2)$, and the intersection has size $\cO(N_1 \cdot N_2)$. We can then find the longest witness (\cref{fct:witness}) of this intersection in time linear in the size of the automaton, that is, in $\cO(N_1 \cdot N_2)$ time.
\end{proof}

Notice that when $T_1$ and $T_2$ are standard strings, then $n_1=m_1=n_2=m_2=1$ and $\cO(N_1\cdot N_2)$ running time matches the conditional lower bound for standard strings up to subpolynomial factors. Hence, no $N_i$ can be replaced by $m_i$ or $n_i$. In particular, an $\cO(N_1m_2+N_2m_1)$ time algorithm would refute the conditional lower bound.

\begin{proposition}\label{prp:space efficient}
In all the solutions presented in this section (as well as in the previous ones) the algorithm can be slightly modified to achieve space complexity $\cO(N_1+N_2)$ without any increase in the running time.
\end{proposition}
\begin{proof}
Notice that all the problems are solved through computing a certain value (reachability, $L(k),S(k),M((i,v))$) for all the explicit nodes of the intersection graph in a topological order. Each time this order is compatible with the order of sets of $T_1$ and $T_2$, and the recurrent formulas refer only to nodes from a single set $U_{i,j}$ or $U'_{i,j}$, that is, the basic computation is enclosed in a single cell of the grid (inspect \cref{fig:grid}). For a single cell of the grid the space used for computation is $\Theta(N_{1,i}+N_{2,j})$, while globally the information passed between the cells is bounded by $\cO(N_1+N_2)$ by going through the cells $(i,j)$ in a lexicographical (or reversed lexicographical) order (it is enough to store information for nodes of $U_{i,j},U'_{i,j}$ only for the next two values of $i$). Finally the size of the input as well as of all the generalized suffix trees is also bounded by $\cO(N_1+N_2)$. 
\end{proof}

\section{(Doubly) Elastic-Degenerate String Matching}\label{sec:DEDSM}

String matching (or pattern matching) on standard strings can be seen as a natural (and very useful) generalization of string equality. Since EDSI is basically a counterpart of the string equality problem (we seek for a pair of standard strings from $\L(T_1)$ and from $\L(T_2)$ that are equal), the problem of generalizing it to string matching arises naturally~\cite{DBLP:journals/iandc/IliopoulosKP21}. In particular, the ED String Matching (EDSM) problem of locating the occurrences of a pattern that is a standard string in an ED text has already been considered~\cite{DBLP:conf/cpm/AoyamaNIIBT18,bernardini_et_al:LIPIcs:2019:10597,elasticSICOMP,DBLP:conf/cpm/GrossiILPPRRVV17,DBLP:journals/iandc/IliopoulosKP21}.

For a standard string $p$ and an ED text $T$, one wants to check whether $p$ occurs in $T$, that is, if it is a substring of some string $t\in \mathcal{L}(T)$ (decision version), or find all the segments $i$ of $T$ such that an occurrence of $p$ starts within a string of the set $T[i]$; i.e., there exists a standard string $t\in\mathcal{L}(T)$ with an alignment such that $p$ occurs in $t$ starting at a position lying in the $i$th segment of the alignment. 

In this section, we consider the EDSM problem in its full generality, i.e., both the text and the pattern are ED strings. We then show how our solution instantiates to the special cases, where one of the two ED strings is a standard string. 

Let us now define the general problem in scope:

\defproblem{Doubly ED String Matching}{An ED string $T$ (called \emph{text}) and an ED string $P$ (called \emph{pattern}).}{YES if and only if there is a string $p\in\mathcal{L}(P)$ that is a substring of a string $t\in\mathcal{L}(T)$ (decision version); or all segments $i$ such that there is a string $p\in\mathcal{L}(P)$ whose occurrence in $T$ starts in $T[i]$ (reporting version).}

Through applying the solution to reversed ED strings $P$ and $T$ (we reverse the sequence of segments and the strings inside them), we can see that the reporting version can be used to find all the segments $i$ where an occurrence ends.
Further note that \textsc{Doubly ED String Matching} is equivalent to asking for the union of the results of EDSM over all the patterns from $\mathcal{L}(P)$ -- we take a binary OR of the results for the decision version.

The algorithms underlying Theorem~\ref{thm:NFA-EDSI} and Theorem~\ref{thm:omega-1 intersection} can both be extended to solve \textsc{Doubly ED String Matching}:

\begin{theorem}
    We can solve the reporting or decision version of \textsc{Doubly ED String Matching} on an ED text $T$ having length $n_T$, cardinality $m_T$ and size $N_T$, and an ED pattern $P$ having length $n_P$, cardinality $m_P$ and size $N_P$ in $\cO(N_Tm_P+N_Pm_T)$ time or in $\cO(N_T^{\omega-1}n_P+N_P^{\omega-1}n_T)$ time.
\end{theorem}
\begin{proof}
Like in \EDSI, we need to check whether an accepting state is reachable from a starting one, on the very same intersection graph. This time, however, for every state $k$ of the path-automaton of $T$, we make every node $(k,1)$ reachable (starting), and every node $(k,n_P+1)$ accepting.
That way, after computing every reachable node with one of the algorithms underlying Theorem~\ref{thm:NFA-EDSI} or Theorem~\ref{thm:omega-1 intersection}, an accepting node is reachable if and only if a string in $\mathcal{L}(P)$ occurs at any position of a string from $\mathcal{L}(T)$. The complexity is the same as before, since the size of graph remains unchanged.

For the reporting version, we observe that occurrences ending in set $T[i]$ correspond to reachable nodes in the set $\{(i,n_P+1)\}\cup U'_{i-1,n_P+1}\backslash \{(i-1,n_P+1)\}$. Hence, we simply report the list of such sets containing a reachable node, which can be done in $\cO(N_T)$ total time over all $i$. The starting positions can be reported through the use of reversed strings.
\end{proof}

In \cite[Corollary 10]{DBLP:conf/wabi/Ascone0CEGGP24} it was shown that \textsc{Doubly ED String Matching} cannot be solved in $\cO(N_PN_T^{1-\varepsilon})$ time or $\cO(N_TN_P^{1-\varepsilon})$ time, for any constant $\varepsilon>0$, unless the OV conjecture is false. This result does not contradict our result as the lengths $n_P,n_T$ can in general be large.

Notice that if the pattern $P$ is a standard string ($n_P=m_P=1$) the transitions from $U'_{i,1}$ to $U'_{i,2}$ are there solely to check if $P$ occurs in any string from $T[i]$ (which can be done in $\cO(N_P+N_T)$ time). Hence, similarly to \cref{cor:one-standard} the running time of our FMM algorithm is equal to $\ctO(n_TN_P^{\omega-1}+N_T)$.
This complexity (and design) matches the result of~\cite{elasticSICOMP}.

Symmetrically, if the text $T$ is a standard string ($n_T=m_T=1$), then the FMM solution runs in $\ctO(n_PN_T^{\omega-1}+N_P)$ time (this time like in case of \cref{cor:one-standard} the edges from $U_{1,j}$ to $U_{2,j}$ are not considered at all).

\section{Approximate \EDSI}\label{sec:AEDSI}

Another popular extension of string equality is \emph{approximate} equality, where one wants to find the distance between two strings: the minimal number of operations that modify one string into the other one, or decide whether this number is at most $k$, for a given integer $k>0$.

We denote by $d_H(S_1,S_2)$ (resp.~$d_E(S_1,S_2)$) the Hamming distance (resp.~edit distance) of two standard strings $S_1,S_2$. The problem of finding the Hamming distance of two standard strings is easily solvable in $\cO(N_1+N_2)$ time, while for edit distance the time increases to $\cO(N_1\cdot N_2)$ time in general case~\cite{DBLP:journals/jacm/WagnerF74} or $\cO(N_1+N_2+k^2)$~\cite{DBLP:journals/siamcomp/LandauMS98} for a given upper bound $k$.

Approximate EDSI gives another measure of similarity when the normal intersection turns out to be empty.

\begin{theorem}
Given a pair $T_1,T_2$ of ED strings of sizes $N_1$ and $N_2$, respectively, we can find the pair $S_1\in \L(T_1),S_2\in \L(T_2)$ minimizing the distance $d_H(S_1,S_2)$ or $d_E(S_1,S_2)$ in $\cO(N_1N_2)$ time.
\end{theorem}
\begin{proof}

We adapt the classic Wagner–Fischer algorithm~\cite{DBLP:journals/jacm/WagnerF74} to our case. We construct the standard (uncompacted) NFA intersection of the two path-automata for $T_1$ and $T_2$, set the weight of all its edges to $0$, and then add extra edges with weight $1$ that represent edit operations.

In case of Hamming distance, when a pair of transitions $u\xrightarrow{l_1} u'$, $v\xrightarrow{l_2} v'$ does not match ($l_1\neq l_2$) we still add the weight-$1$ edge between $(u,v)$ and $(u',v')$ ($l_1$ is substituted with $l_2$).

In case of edit distance we additionally construct weight-$1$ edges from $(u,v)$ to $(u',v)$ for every transition $u\xrightarrow{l_1} u'$ (corresponding to deletion of the letter $l_1$), and to $(u,v')$ for every transition  $v\xrightarrow{l_2} v'$ (corresponding to insertion of letter $l_2$).

Now all we need to do is to find the minimum cost path from the starting node to the accepting one (in case of Hamming distance only if such a path exists), which can be done in linear time using Dial's implementation of Dijkstra's algorithm~\cite{DBLP:journals/cacm/Dial69}.
By simply spelling the labels of the found path we obtain strings $S_1$ and $S_2$ (the difference between them is encoded through the labels of the weight-$1$ edges).
\end{proof}

Without having a threshold $k$, even when $T_1,T_2$ are standard strings ($n_1=n_2=m_1=m_2=1$), for the edit distance we cannot obtain an algorithm running in $\cO((N_1N_2)^{1-\varepsilon})$ time for a constant $\varepsilon>0$ unless SETH fails~\cite{DBLP:journals/corr/BackursI14}, or a faster algorithm through a parameterization with $n_1,n_2,m_1,m_2$.

On the other hand, if we are given a threshold $k$, we can adapt another classic algorithm, given by Landau and Vishkin~\cite{Landau1986-vm}, to obtain the following theorem.

\begin{theorem}
Given an ED string $T_1$ of cardinality $m_1$ and size $N_1$, an ED string $T_2$ of cardinality $m_2$ and size $N_2$, and an integer $k>0$, we can check whether a pair $S_1\in \L(T_1),S_2\in \L(T_2)$ with $d_H(S_1,S_2)\le k$ (resp. $d_E(S_1,S_2)\le k$) exists in $\cO(k(N_1m_2+N_2m_1))$ time (resp. in $\cO(k^2(N_1m_2+N_2m_1))$ time) and, if that is the case, return the pair with the smallest distance.
\end{theorem}
\begin{proof}
This time we make use of the compacted intersection automaton from \cref{lem:NFA-intersection}.
In addition to standard weight-$0$ edges, which are constructed when $\LCP(l_1,l_2)=\min(|l_1|,|l_2|)$ for $l_1,l_2$ being the labels of extended transitions in the two path-automata, we also add new edges when one of the strings is at distance at most $k$ from a prefix of the other one.
More formally, for a pair of extended transitions $u\xrightarrow{l_1} u'$, $v\xrightarrow{l_2} v'$, if $l_1$ is at distance $k'$ from the length-$x$ prefix of $l_2$, we produce a weight-$k'$ edge from $(u,v)$ to $(u',v_x)$, where $v_x$ is the implicit state between $v$ and $v'$ representing the length $x$ prefix of $l_2$ (symmetrically for $l_2$ at distance at most $k$ from a prefix of $l_1$). 
All such values of $k'\le k$ can be found with the use of a standard LCP queries data structure and \emph{kangaroo jumps}~\cite{10.1145/8307.8309} in $\cO(k)$ time for Hamming distance and in $\cO(k^2)$ time for edit distance~\cite{Landau1986-vm}.

In case of Hamming distance, the total number of extended transitions is still $\cO(N_1m_2+N_2m_1)$ (the number of transition checks does not change). In case of edit distance, it can increase by a factor of $k$. Indeed, for a single pair of transitions $l_1,l_2$, we can produce up to $2k+1$ weighted transitions.

After that, once again, we can use Dial's implementation of Dijkstra's algorithm~\cite{DBLP:journals/cacm/Dial69} to find the smallest weight path in this graph, obtaining the claimed result.
\end{proof}

One may notice that the modifications to the standard intersection algorithms from this section and the previous one are independent; in the previous section those consisted of marking additional nodes as starting/accepting, while in this section those consisted of adding edges.
By applying both modifications simultaneously, we can solve the Approximate (Doubly) EDSM problem, that is, the problem of finding a pair of strings $S_P\in\mathcal{L}(P),S_T\in\mathcal{L}(T)$ such that $S_P$ is at distance at most $k$ from a substring of $S_T$ (for which this distance is minimized). We thus obtain the following result:

\begin{corollary}
Approximate (Doubly) EDSM can be solved in $\cO(N_1N_2)$ time in general, or in $\cO(k(N_1m_2+N_2m_1))$ time for Hamming distance and $\cO(k^2(N_1m_2+N_2m_1))$ time for edit distance when we are given an integer $k>0$ as an upper bound on the sought distance.
\end{corollary}

Once again we can notice that when pattern $P$ is a standard string this time complexity can be bounded by a smaller value -- we can check if the pattern $P$ appears in any of the strings from any of the sets $T[i]$ in $\cO(k(N_P+N_T))$ total time~\cite{Landau1986-vm}. Otherwise each edge used has at least one endpoint in a node $(k,j)$, where $j$ is an explicit state from the automaton representing $T$, hence the total number edge checks in this case is equal to $N_Pm_T$ (we can compute the backwards edges the same way as the forward ones), thus we obtain the total running time of $\cO(k(N_Pm_T+N_T))$ for Hamming distance and $\cO(k^2N_Pm_T+kN_T)$ time for edit distance. Our results match the time complexity of the corresponding results of~\cite{tcs-ed2020}. (Note that~\cite{tcs-ed2020} uses $G$ to denote $m_T$ and $m$ to denote $N_P$.)

Symmetrically, when text $T$ is a standard string those time complexities are $\cO(k(N_Tm_P+N_P))$ and $\cO(k^2N_Tm_P+kN_P))$, respectively, since in this case every useful edge has at least one endpoint in a node $(i,k)$, where $i$ is an explicit state from the automaton representing $P$.

\section{Open Questions}\label{sec:fin}

In our view, the main open questions that stem from our work are as follows.
\begin{itemize}
\item We showed an $\ctO(n_2N_1^{\omega-1}+n_1N_2^{\omega-1})$-time algorithm for $\EDSI$.
Can one design an $\cO(n_2N_1^{\omega-1-\epsilon}+n_1N_2^{\omega-1-\epsilon})$-time  (perhaps not combinatorial) algorithm for \EDSI, for some $\epsilon>0$?
\item We showed that there is no combinatorial $\cO((N_1+N_2)^{1.2-\epsilon}f(n_1,n_2))$-time algorithm for \EDSI. Can one  show a stronger conditional lower bound for combinatorial algorithms?
\item We showed an $\cO(N_1\log N_1\log n_1+N_2\log N_2\log n_2)$-time algorithm for outputting a representation of the intersection language of two unary ED strings. Can one design an
$o(N_1\log N_1\log n_1 + N_2\log N_2\log n_2)$-time algorithm?
\end{itemize}

\bibliographystyle{plain}
\bibliography{arxiv}

\end{document}